\documentclass[10pt, a4paper]{article}

\usepackage{pstricks}
\usepackage{algorithm}
\usepackage[noend]{algpseudocode}
\usepackage{tikz}
\usepackage{comment}
\usepackage{geometry}
\usepackage[T1]{fontenc}
\usepackage[utf8]{inputenc}
\usepackage{authblk}
\usepackage[english]{babel}
\usepackage{hyperref}

\usepackage{amsfonts}
\usepackage{amsthm}
\usepackage{amsmath}
\usepackage{amssymb}
\usepackage{multirow}
\usepackage{arydshln}

\newtheorem{theorem}{Theorem}

\newtheorem{lemma}{Lemma}

\allowdisplaybreaks

\def\Z{\mathbb Z}

\usepackage[numbers]{natbib}   
\makeatletter
\def\BState{\State\hskip-\ALG@thistlm}
\makeatother

\title{The Bounded Beam Search algorithm for the Block Relocation Problem}

\author[1]{Tiziano Bacci}
\author[1]{Sara Mattia}
\author[1]{Paolo Ventura}
\affil[1]{IASI - CNR, Via dei Taurini 19, 00185 Roma, Italia \texttt{\{tiziano.bacci,sara.mattia,paolo.ventura\}@iasi.cnr.it}}

\date{\today}

\begin{document}

\maketitle

\begin{center}
\fbox{\begin{minipage}{35em}
This is the postprint version of the article
``Bacci, T. , Mattia, S. , Ventura, P. (2019). The bounded beam search algorithm for the block relocation
problem. Computers \& Operations Research, 103 , 252–264 .''
The published version is available at https://doi.org/10.1016/j.cor.2018.11.008.
\end{minipage}}
\end{center}

\begin{abstract}
\noindent 
In this paper we deal with the restricted Block Relocation Problem. We
present a new lower bound and a heuristic approach for the problem. 
The proposed lower bound can be computed in polynomial time and
it is provably better than some previously known lower bounds.
We use it within a bounded beam search algorithm to solve the Block Relocation Problem
and show that the considered heuristic approach outperforms the other existing
algorithms on most of the instances in the literature. 
In order to test the approaches on real-size dimensions,  
new large instances of the Block Relocation Problem are also introduced.\\


\end{abstract}

\section{Introduction}
\label{intro}

Let $\cal S$ be a system (yard) defined by $w$ stacks of capacity (in terms of available slots/tiers) $h$ and
let $\{1,\dots,n\}$ be a set of $n$ blocks located in the slots of the $w$ stacks.
A {\em reshuffle operation} (or simply a {\em reshuffle} or a {\em relocation}) is a movement of a block from a stack to another,
while a {\em retrieval} is a movement of a block from a stack to the outside of the system.
The stacks can store blocks according to a last-in / first-out policy.
In a stack, only the topmost block is accessible and, when a block has to be retrieved,
all the blocks above it have to be reshuffled. 
When a block is reshuffled and moved in one of the other stacks of the yard,
it has to be allocated in the first slot available from the bottom to the top.

The Block Relocation Problem (BRP) consists in deciding where to reallocate every block that is moved by a reshuffle operation,
in order to minimize the total number of reshuffles needed to retrieve all the blocks
according to the retrieval order (1,...,n).
Observe that such minimum is lower bounded by the number of {\em blocking} blocks
of $\cal S$, i.e., those located in any slot above a block with higher retrieval priority.
Figure \ref{exBRP} gives an example of the BRP with a system defined by $w = 3$ stacks, $h = 3$ available slots for each stack,
and $n = 6$ blocks. Starting from the initial yard, 
where blocks 5 and 6 are blocking, the sequence of movements of an optimal solution is reported.
At each step, the next block to be moved with a reshuffle or a retrieval operation is highlighted in gray.
The minimum number of reshuffles required is three: block 6 is reshuffled in order to retrieve block 1; block 5 is reshuffled to retrieve block 2; we then reshuffle block 6 to retrieve block 4.
Since BRP generalizes the Bounded Coloring Problem (also known as Mutual Exclusion Scheduling)
on permutation graphs, it is known to be NP-hard for any fixed $h \geq 6$. 
For some complexity results on BRP, Bounded Coloring and related problems,
see  \citet{CSV2012}, \citet{J2003}, \citet{BMO2011}, and \citet{BMV2017}. 
 
A real world application of the BRP arises in the logistics of containers in a terminal. A container terminal is an area where containers are stored and transshipped between different transport vehicles, such as cargo ships, trains, trucks, and where they are stacked because of the limited storage space.
The storage area (yard) is often divided into groups of stacks of containers,
called bays, and containers are moved by yard cranes. The movements of containers may occur  within the same bay or between different bays (\citet{LL2010}). Typically, a container yard stores at the same time thousands of containers grouped into hundreds of stacks with a storage capacity which may be up to 10 slots (\citet{GK2005}). Since a stack is accessible only from the top, when a container is required outside of the storage area, any container located above it has to be moved to another stack by a yard crane with a reshuffle operation. Reshuffle operations are time-consuming and they have to be avoided as much as possible. In this scenario, the Block Relocation Problem consists in finding a way to retrieve, in a given order, all the containers in the container yard so to minimize the number of reshuffle operations.
Throughout the paper, the words container and block will be used interchangeably.

In the literature, two variants of the BRP are studied: {\em restricted} and {\em unrestricted}.
In the restricted version, it is allowed to reshuffle only blocks located above the next one to be retrieved,
while, in the unrestricted case, any block can be reshuffled.
In this work, we focus on the first variant.
Readers interested in the unrestricted version can refer to
\citet{FB2012}, \citet{PH2013}, \citet{TM2015}, and \citet{TSB2018}.

The best exact approaches proposed for the restricted BRP allow to solve instances with up to 70 containers (\citet{ZQLZ2012}
and \citet{TT2016}), that is, in most cases, far away from the dimension of a real world scenario (\citet{LL2010}). 
Therefore, a huge amount of heuristics have been proposed in the literature.
Here we describe a new beam search approach, the Bounded Beam Search algorithm, 
and we compare it with the other methods in the literature, showing that it outperforms most of the existing approaches.
In the Bounded Beam Search algorithm we make use of a new lower bound for BRP, denoted as the Unordered Blocks Assignment Lower Bound.  
We compare it with other lower bounds in the literature and prove that it theoretically dominates the two most used ones. 
In addition, we introduce a new set of instances, large enough to represent real scenarios.

In Section \ref{literature}, we give a survey of the literature on BRP, focusing on the heuristic procedures (Section \ref{sec:heuralgo})
and on the lower bounds (Section \ref{sec:extlb}); some variants of the BRP are described in Section \ref{sec:relpr}.
In Section \ref{sec:newlowerbound}, we introduce the Unordered Blocks Assignment
Lower Bound.
In Section \ref{sec:beamserch}, we describe the Bounded Beam Search heuristic algorithm in details.
In Section \ref{instances}, we describe the test bed used in the computational experiments.
In Section \ref{helb}, we study how the behavior of the Bounded Beam Search
algorithm changes according to some parameters.
In Section \ref{sec:results}, we compare the performances of the Bounded Beam
Search procedure with the best heuristic methods taken from the literature on
different sets of instances.
Finally, in Section \ref{conclusions}, we present our comments and conclusions.


\begin{figure}[htbp] 
  \begin{center}
\includegraphics[width=15cm,height=6cm]{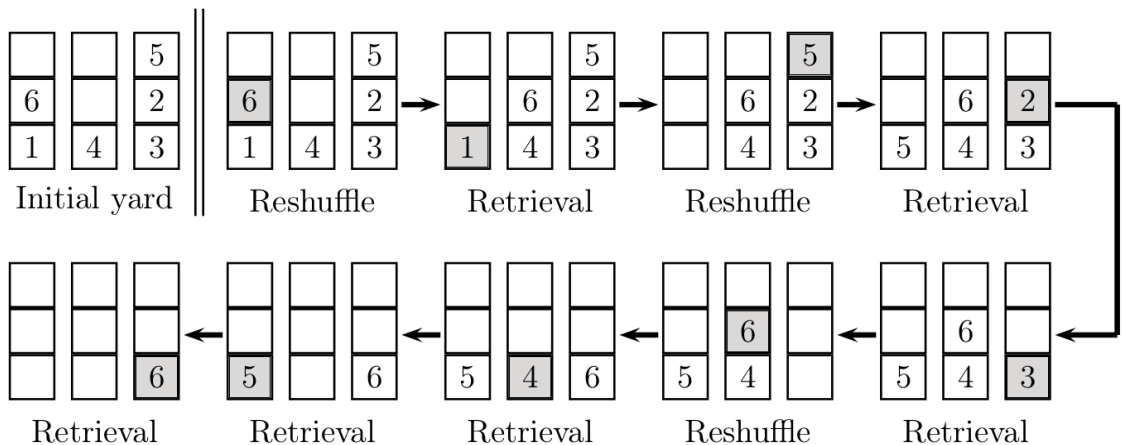}
    \end{center}
\caption{A representation of an optimum solution for the Block Relocation Problem.}
\label{exBRP}
\end{figure}


\section{Literature review}
\label{literature}

A huge amount of literature concerning storage yard operations in container terminals
has been proposed in the last two decades, as described in the recent surveys by
\citet{CVR2014} and \citet{LK2014}. Below we review the literature for the BRP and
some of its variants, focusing on heuristic approaches and lower bounds for the problem.

\subsection{Exact approaches}

The exact approaches for the problem can be distinguished into methods based on
Integer Linear Programming (ILP) formulations (\citet{CSV2012}, \citet{WLT2009},
\citet{ZF2014}, \citet{ZCFSV2015}) and search-based methods
(\citet{KH2006}, \citet{WT2010}, \citet{UA2012},  \citet{ZQLZ2012}, \citet{IBV2014},
\citet{TT2016}, \citet{IBV2015}, \citet{KA2016b}).
ILP based methods can solve only instances up to 30 blocks, while search-based
exact algorithms can tackle instances with up to 70 blocks (see \citet{TT2016}). 

\subsection{Heuristic algorithms}
\label{sec:heuralgo}

Many heuristic procedures have been proposed for the BRP.
These algorithms can be grouped into {\em fast} or {\em slow} methods.
In the first group we include all the algorithms 
based on a simple greedy approach, while the second group contains procedures that are more
sophisticated or structured.
The fast approaches include the following algorithms: 
the Lowest Position heuristic (\citet{Z2000}); 
the Reshuffle Index heuristic (\citet{MLTLLC2005});
the heuristic rule based on the Expected Number of Additional Relocations (\citet{KH2006});
the Extended Lowest Position heuristic, the Extended Reshuffle Index heuristic,
the Extended heuristic based on the Expected Number of Additional Relocations (\citet{WLT2009});
the Reshuffle Index with Look-ahead heuristic (\citet{WT2010});
the heuristic procedure proposed in \citet{CSV2012};
the Lowest Absolute Difference heuristic, the Group Assignment Heuristic by \citet{WT2012};
the Greedy1 heuristic, the Difference1 heuristic (\citet{UA2012});
the Probe Restricted 1 heuristic, the Probe Restricted 2 heuristic, the Probe Restricted 3 heuristic, the Probe Restricted 4 heuristic (\citet{ZQLZ2012});
the Chain heuristic, the ChainF heuristic (\citet{JV2014}).
Below we briefly describe such fast procedures. 

In the Lowest Position heuristic the next block that has to be reshuffled is assigned to the stack that has the highest number of empty slots.
The Reshuffle Index heuristic assigns a reshuffled block to the stack with the minimum number of containers that have an higher priority (a lower index). 
The heuristic rule based on the Expected Number of Additional Relocations assigns the next block to be reshuffled to the stack with the minimum expected number of additional relocations. 
The Extended Lowest Position heuristic, the Extended Reshuffle Index heuristic, the Extended heuristic based on the Expected Number of Additional Relocations  are modified (refined) versions of the Lowest Position heuristic, the Reshuffle Index heuristic and the heuristic rule based on the Expected Number of Additional Relocations, respectively.
In particular, the novelty consists of assigning each reshuffled block to the stack that minimizes an
upper bound computed using the original heuristic.
the Reshuffle Index with Look-ahead heuristic is also a development of the Reshuffle Index procedure and they differ in the rule for breaking ties.
In the heuristic procedure proposed in \citet{CSV2012} a block $i$ is assigned to the stack $s = arg \min \{\sigma_s| \sigma_s > i\}$;
otherwise to the stack $s = arg \max \{\sigma_s\}$.
Here and in the following, $\sigma_s$ is the block with the minimum index among those allocated in stack $s$. 
We note that this procedures, the Lowest Absolute Difference heuristic and the Probe
Restricted 3 heuristic correspond to the same algorithm.
To move the reshuffled container $i$ in a stack $s$, Difference1 considers (in this order) the following
three criteria:
i) $s = arg \min \{\sigma_s | \sigma_s > i\}$;
ii) $s = arg \min \{k_s - i | k_s < i\}$, where $k_s$ is the topmost block of $s$;
iii) $s = arg \min \{k_s\}$.
The Probe Restricted 1 heuristic and the Lowest Position procedure differ only in the rule adopted for breaking ties.
In the Probe Restricted 2 heuristic, the next block to be reshuffled is assigned to the stack that maximizes $\sigma_s$.
As in Difference1, also in the Probe Restricted 4 heuristic the stack $s$ is selected
iteratively applying different criteria on $\sigma_s$. 
Differently from the approaches described so far, the Group Assignment Heuristic, ChainF, Chain, and Greedy1 consider, at each step,
not only the next block that has to be reshuffled, but a subset of the blocking blocks $Q$
above the next one to be retrieved. 
Chain and ChainF iteratively reshuffle the block of $Q$ in consecutive pairs, whereas the Group Assignment Heuristic and
Greedy1, assign all the blocks in $Q$ at the same time.  

The slow approaches include the following algorithms:
the Minimum Reshuffle Integer Program heuristic (\citet{WLT2009});
the Tabu Search heuristic by \citet{WTH2010}; 
the Beam Search heuristic by \citet{WT2010};
the Corridor Method heuristic (\citet{CVS2011});
the Matrix-Algorithm heuristic (\citet{CSV2009});
the three phase-heuristic by \citet{LL2010};
the iterative deepening A* restricted heuristic (\citet{ZQLZ2012});
the depth-first branch \& bound heuristic (\citet{KA2016b}).

The Minimum Reshuffle Integer Program heuristic is based on an Integer Linear Programming formulation, where a parameter $k$ indicates that such model minimizes the total number of reshuffles in retrieving the first $k$ containers. 
The Tabu Search heuristic implements a tabu search approach.
The Beam Search heuristic is a beam search algorithm that uses the Reshuffle Index heuristic to compute the upper bounds. 
The basic idea of Corridor Method is to use an exact method over restricted portions of the solution space in order to find good feasible solutions. 
\citet{ZQLZ2012} noted that it may provide infeasible solutions.
The Matrix-Algorithm explores all the configurations obtained by performing one reshuffle and by assigning to each child a score computed by a fast heuristic. 
Then, it randomly selects the next node to be explored. 
Its performances benefit from a smart binary encoding of the stacking area. 
The three phase-heuristic operates according to three phases: initially, a first feasible solution is obtained by a greedy procedure; then, an ILP formulation is used to try reducing the number of reshuffles; finally, another ILP is used to reduce the total time. 
The iterative deepening A* restricted heuristic is an implementation of the iterative deepening A* approach, applied to the restricted BRP. It is an exact method, that can be turned in a heuristic one, by setting a time limit.  
The depth-first branch \& bound heuristic is obtained by fixing a time limit to an exact search based algorithm, which uses an abstraction method to reduce the solution space. 

\subsection{Lower bounds}
\label{sec:extlb} 

There are few contributions in the literature related to lower bounds for the restricted BRP. \citet{KH2006} introduce a lower bound (referred in the literature as LB1, see \citet{ZQLZ2012}) which derives from the simple observation that each blocking container has to be reshuffled at least once. 

\citet{ZQLZ2012} present a lower bound (LB3) that can be calculated according to
the following iterative procedure.
Set $Q = 0$. 
Then, for $i=1, \ldots, n$, consider container $i$ that has to be retrieved and
let $B^i$ be the set of containers blocking $i$.
For each of these containers that will necessarily be blocking after its
relocation, the variable $Q$ is increased by 1.
The yard at iteration $i+1$ is obtained by removing all the blocks in
$B^i \cup \{i\}$. Then, the value of LB3 is LB1 + $Q$.

The lower bound LB4, introduced by \citet{TT2016}, is a refinement of LB3. 
However, while LB1 and LB3 can be calculated in polynomial time,
the algorithm for computing LB4 is exponential in the number of blocks. 

\subsection{Related problems}
\label{sec:relpr} 

The {\em Container Pre-Marshalling Problem} (CPMP, \citet{LH2007}) has several common aspects with the BRP. As in the latter, in the CPMP a set of containers, each of them having a retrieval priority, are located in a yard composed by a set of stacks with a given maximum height. The difference with BRP is that the blocks are not removed
from the yard and only reshuffle operations are allowed. Here one has to reorder the blocks of the initial configuration so to obtain a yard without blocking blocks, with
a minimum number of reshuffles. The CPMP has been largely studied in the literature and several heuristic algorithms and exact approaches have been presented (\citet{HT2016,TT2018,JTV2017}).
Another variant of BRP is the {\em Blocks Relocation Problem with Batch Moves} where more than one container at a time can be retrieved from the yard or moved between two stacks (\citet{ZLK2016}).  

Several robust variants of BRP are also studied: they include the {\em Online Container Relocation Problem} (\citet{ZFJ2017}), in which the retrieval order is known only for a subset of the $n$ containers allocated in the initial configuration and the {\em Container relocation problem with time windows} (\citet{KA2016}), where the containers are divided into groups and it is given a retrieval order among the groups, but not among the blocks belonging to the same group.

\section{A new lower bound for the BRP}
\label{sec:newlowerbound}

In this Section, we present a new lower bound, called Unordered Blocks Assignment
Lower Bound (UBALB), for the restricted BRP.
A preliminary version of this contribution can be found in \citet{BMV2018}.

In the following we show that such a lower bound can be computed by iteratively
solving a relaxation of the Generalized Blocking Items Problem (GMBIP) defined
below.
Given a yard $M \in {\Z}^{w \times h}$ composed by $w$ stacks of height $h$, let $M(j,k)$ be the block located in the $k$-th slot of stack $j$
($M(j,k) = 0$ if the position is empty). Suppose to have a set $B$ of $m$ blocks that have to be allocated in the yard, according to a given order $\phi$ (denote by $\phi(i)$ the $i$-th block to be located) and let
$\overline M$ be the yard obtained from $M$ after placing these blocks. 
$\overline M$ is said to be $\phi$-{\em feasible} if it is compatible with
$\phi$, i.e., for each couple $b, b' \in B$ with $\phi^{-1}(b) < \phi^{-1}(b')$,
$b$ is not located above $b'$.
Furthermore, a block $r \in B$ is \emph{$r'$-blocking} if it is located above
block $r' < r$ and it is \emph{blocking} if it is $r'$-blocking for some $r'$.
Given an input instance defined by ($M$, $B$, $\phi$), the GMBIP consists of
finding a $\phi$-feasible $\overline M$ that minimizes the total number of
blocking blocks. The optimal value is denoted by $G^\ast (M, B, \phi)$.

\begin{lemma} \label{lemma:GMBIP}
GMBIP is NP-hard. 
\end{lemma}
\begin{proof}
It contains the Minimum Blocking Items Problem (MBIP) as special case. 
Indeed, in MBIP the initial yard $M$ is assumed to be empty. 
MBIP is known to be NP-hard (\citet{BMV2017}).
\end{proof}

The solution of GMBIP can be used to obtain a lower bound for BRP, as we explain in the following. 

Given a set $\{1,\dots,n\}$ of $n$ blocks located in a yard $M=M^0$, let $M^i$ be obtained from $M^{i-1}$ by removing the block $i$ located in stack $t^i$ and the set $B^i$ of $i$-blocking blocks (let $\phi^i$ be the order given by taking them from the top to the bottom), for each $i=1,\dots,n$. It is worth observing that: (i) $M^n$ is the empty yard; (ii) block $i$ may not be present in $M^{i-1}$; and (iii) each block in $B^i$ has to be reshuffled at time $i$. Furthermore, a block in $B^i$ may become blocking again after its reshuffle. The minimum number of such blocks is $G^*(\tilde M^{i}, B^i, \phi^i)$,
where $\tilde M^{i}$ is obtained from $M^{i}$ by removing stack $t^i$. 

A lower bound for the BRP instance defined by $M$ is then given by the following equation:

\begin{equation}
\label{eq:lb}
\sum\limits_{i=1}^{n} (|B^i| +  G^*(\tilde M^{i}, B^i, \phi^i))
\end{equation}

Since GMBIP is NP-hard, the lower bound given by equation \eqref{eq:lb} is hard to compute in general.
All the lower bounds for the BRP presented in the
literature are, indeed, derived from \eqref{eq:lb} substituting, at each iteration $i$,
$G^*(\tilde M^{i}, B^i, \phi^i)$ with some lower bound.

\citet{KH2006} compute LB1 by setting $G^*(M, B, \phi) = 0$.
For calculating LB3, \citet{ZQLZ2012} relax both the order $\phi$ and the capacity $h$ of the stacks. 
\citet{TT2016} only relax the capacity restriction on the stacks.
We denote the optimal values of the corresponding relaxed problems by $G^Z(M, B, \phi)$ and $G^T(M, B, \phi)$ respectively.
Note that the algorithm presented by \citet{ZQLZ2012} to calculate $G^Z(M, B, \phi)$ runs in $O(w+n)$, whereas the one by \citet{TT2016} for $G^T(M, B, \phi)$ requires $O((w + 2^n)log(w))$ steps.
In the following, we say that a certain lower bound A {\em dominates}
another lower bound B if, for every input instance, the value
of A is not less than the value of B.
Since $0 \leq G^Z(M, B, \phi) \leq G^T(M, B, \phi)$ it holds that: (i) LB3 dominates  LB1; (ii) LB4 dominates LB3.

Here we introduce a new lower bound based on a relaxation GMBIP$^B$ of GMBIP, 
where the final matrix ${\overline M}$ does not need to be $\phi$-feasible. 
Let $G^B(M, B, \phi)$ be the optimal value of GMBIP$^B$. 

\begin{theorem} \label{th:dominanza}
  The following holds: (i) the Unordered Blocks Assignment
Lower Bound dominates  LB3;
  (ii) neither LB4 dominates the Unordered Blocks Assignment
Lower Bound nor the opposite.
\end{theorem}
\begin{proof}

\noindent 
(i). Trivially, for any $i$, $G^Z(\tilde M^{i}, B^i, \phi^i)) \leq G^B(\tilde M^{i}, B^i, \phi^i))$.
In fact, the former relaxes both the $\phi$-feasibility and constraints on the height of the stacks, whereas the latter only relaxes the $\phi$-feasibility.\\
(ii). Figure \ref{fig:gb_gt} shows an instance of BRP where 
LB4 has a value that is strictly less than the one of the Unordered Blocks Assignment
Lower Bound,
while Figure \ref{fig:gt_gb} represents an instance where the opposite relation
holds.
The figures depict the initial yard of the corresponding BRP problem on the left side.
In both cases, it is easy to see that $G^T(\tilde M^{i}, B^i, \phi^i) = G^B(\tilde M^{i}, B^i, \phi^i)$ for any $i \neq 1$.
Instead, in the first case we have $G^T(\tilde M^{1}, B^1, \phi^1) < G^B(\tilde M^{1}, B^1, \phi^1)$, whereas in the second case $G^T(\tilde M^{1}, B^1, \phi^1) > G^B(\tilde M^{1}, B^1, \phi^1)$.
The corresponding computation is reported on the right side of the figures.
As a consequence, in the first example, the value of LB4 is strictly less
(strictly greater, in the second case) than the one of the Unordered Blocks Assignment
Lower Bound.
Notice that, in Figure \ref{fig:gb_gt}, the optimal configuration $\overline M_{T}^1$ for $G^T(\tilde M^{1}, B^1, \phi^1))$, exceeds the capacity of the first stack.
\end{proof}


\begin{figure}[htbp] 
  \begin{center}
\includegraphics[width=10cm,height=3cm]{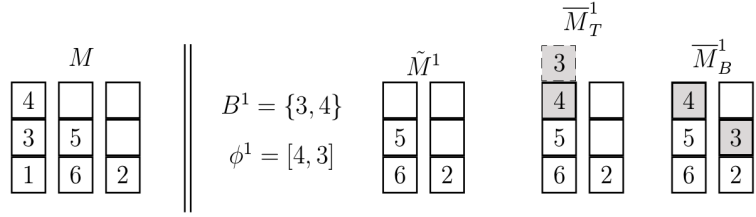}
  \end{center}
  \vspace{-.5cm}
  \caption{An example where the value of LB4 is strictly less than the
    one of the Unordered Blocks Assignment
Lower Bound.}
  \label{fig:gb_gt}
\end{figure}
        
\begin{figure}[htbp] 
  \begin{center}
\includegraphics[width=10cm,height=2.5cm]{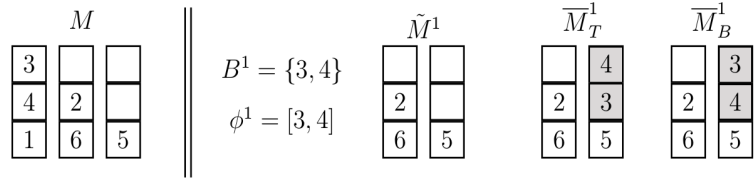}
  \end{center}
  \vspace{-.5cm}
  \caption{An example where the value of LB4 is strictly greater than the
    one of the Unordered Blocks Assignment
Lower Bound.}
  \label{fig:gt_gb}
\end{figure}


In the following we present an algorithm for solving $GMBIP^B(M, B, \phi)$.

Let $B = \{b_1,...,b_n\}$, $\delta_j$ be the number of available empty slots of stack
$j$ of $M$ and let $\sigma_j$ be the minimum index of a block located in $j$. 

\medskip
{\bf Algorithm 1}

\medskip
\begin{algorithmic}[1]
\State set G$^B$ = 0
\State order the blocks of $B$ such that $b_1 > b_2>...>b_n$
\State order the stacks so that $\sigma_1 > \sigma_2>...>\sigma_w$
\For{$i = 1, \dots, n$}
	\State let $R = \{j' \in \{1,...,w\}$ | $\sigma_{j'} > b_i$, $\delta_{j'} > 0\}$
    \If{ $R \neq \emptyset$}
    	\State let $j$ = argmax $\{\sigma_{j'} | j' \in R\}$.
    \Else	
	\State let $j$ = argmin $\{\sigma_{j'}\}$ for all $j'$ with $\delta_{j'} > 0$
	\State G$^B$ = G$^B$ + 1
    \EndIf
    \State locate $b_i$ in $j$
    \State $\delta_j = \delta_j-1$
\EndFor
\end{algorithmic}

GMBIP$^B$ can be seen as a minimum cost assignment problem, where each block
$b_i \in B$ has to be assigned to some stack $j$ with $\delta_j > 0$.
For each $i \in \{1,...,n\}$ and $j \in \{1,...,w\}$, assigning $i$ to $j$ has 
cost 1, if $b_i > \sigma_j$, and 0 otherwise.
Given $B$, $w$, $\delta$, and $\sigma$,  an instance of such assignment problem
is denoted by ASS$(B,w,\delta, \sigma)$.

\begin{theorem}
Algorithm 1 solves GMBIP$^B$ in $O(n \log(n) + w \log(w))$.
\end{theorem}

\begin{proof}
  Given an input instance defined by a quadruple $(B,w,\delta, \sigma)$,
  we first prove that Algorithm 1 produces an optimal solution to the equivalent problem   ASS$(B,w,\delta, \sigma)$.
  In the following, let $\overline X \in \{0,1\}^{n \times w}$ and $X^* \in \{0,1\}^{n \times w}$
  denote the solution provided by Algorithm 1 and an optimal solution to
  ASS$(B,w,\delta, \sigma)$, respectively.
  According with this notation, $\overline X(i,j) = 1$ ($X^*(i,j) = 1$) indicates that
  block $b_i$ is located in stack $j$ ($i$ is assigned to $j$, respectively).
  Assume that ${\overline X}$ is not $X^*$ and let $i$ be smallest index such that ${\overline X}(i,j)$ and $X^*(i,j') = 1$ with $j \neq j'$. 
  We will show that we can produce another optimal solution where $b_i$ is 
  assigned to $j$. Indeed, let $X'$ be such that ${X'}(i,j) = 1$, 
  ${X}'(i,j') = 0$, and ${X}'(\bar i, \bar j) = X^*(\bar i, \bar j)$
  for all the other entries $(\bar i, \bar j) \notin \{(i,j),(i,j')\}$.
  Then two cases can occur. Case a): $X'$ is feasible. Case b): $X'$ is infeasible. 
  In the latter case, since $X^*$ is feasible, it follows that, in $X'$,   stack $j$ exceeds its capacity $\delta_j$. 
  Then observe that, since ${\overline X}(i,j) = 1$, the number of blocks in $b_1,...,b_{i-1}$ that are assigned to stack $j$ (both in ${\overline X}$ as well as in $X^*$) is strictly less that $\delta_j$. 
  Therefore, there exists $i' > i$ such that $X'(i',j) = 1$. 
  Hence,  we can construct a feasible matrix $\tilde X$ from $X'$, by setting 
  $\tilde X(i',j) = 0$ and $\tilde X(i',j') = 1$.
  Now, in order to prove that $X'$ (Case a)) or $\tilde X$ (Case b)) is still optimal, we distinguish two sub-cases.
  First consider $C(i,j) = 0$. Then, for Case a), it follows that $X'$ is
  optimal. For Case b), observe that, as $i' > i$, then $b_{i'} < b_{i}$ and,
  consequently, $C(i',j') \leq C(i,j')$. 
  Therefore, $C(i,j) + C(i',j') \leq C(i,j') + C(i',j)$ and $\tilde X$ is optimal.
  Now assume $C(i,j) = 1$. Observe that $C(i,j') = 1$, otherwise, for steps 6-7  of 
  Algorithm 1, $b_i$ would have been assigned to stack $j'$ instead of $j$.
  This proves that, in Case a), $X'$ is optimal. For Case b), notice that
  $\sigma_{j'} > \sigma_j$, otherwise, for steps 8-9 of Algorithm 1, again
  $b_i$ would have been assigned to stack $j'$. As a consequence, 
  $C(i',j') \leq C(i',j)$. Therefore, also in this case, 
  $C(i,j) + C(i',j') \leq C(i,j') + C(i',j)$ and $\tilde X$ is optimal.  

  Finally, we analyze the computational complexity of the algorithm. 
  Ordering the blocks of $B$ requires $O(n\log(n))$ steps, while the stacks can be ordered in $O(w\log(w))$. 
  Then, loop 4-12 is repeated $n$ times and, at each iteration, stack $j$ can be found 
  in constant time.
  Hence the overall complexity is $O(n \log(n) + w \log(w))$.
\end{proof}

\section{A beam search algorithm for the BRP: the general scheme}
\label{sec:beamserch}

Beam Search is a classic meta-heuristic based on an iterative enumeration of the feasible solutions that are represented by the leaves of an enumeration tree.
Such an enumeration is not complete because the number of descendants of each node is bounded by a parameter $\beta$.
The main ingredients of the approach are then: an algorithm to calculate an upper bound UB(Y) for the value of the best solution
that can be obtained by each node Y of the tree, and a value for the parameter $\beta$.
Therefore, for each node of the enumeration tree, the algorithm generates only the 
best (according to the value UB(Y)) $\beta$ descendants.

Here we present an improved beam search algorithm, called Bounded Beam Search, for the BRP. With respect to the standard 
beam search framework, we also make use of an algorithm to calculate a lower bound 
LB(Y) for the value of the best solution associated with each node Y. 
Such a value will be used both to select the best descendants to be generated and to further restrict the search space.
In the following, we will describe the Bounded Beam Search procedure in more details. 
Algorithm 2 reports the pseudocode of the algorithm.

\medskip
{\bf Algorithm 2}
\medskip
\begin{algorithmic}[1]
		\State $i = 0$
		\State $Y_0 =$ initial yard
		\State $CUB = UB(Y_0)$
		\State $S_0 = \{ Y_0 \}$
	\While{$S_i \neq \emptyset$}
		\For{$Y \in S_i$}
			\Procedure{Retrieval}{}
				\While{the next block to be retrieved $r$ is not blocked}
					\State remove $r$ from $Y$
				\EndWhile
			\EndProcedure
			\If {$Y \neq \emptyset$} 
				\State Let $u$ be the blocking block of $Y$ that has to be reshuffled as first	
				\Procedure{Level construction}{}
					\State $Q = \emptyset$
					\For{$s = 1, \dots, w$}
						\If {stack $s$ is not full} 
						\State Let $\overline Y$ be obtained from $Y$ by moving $u$ in stack $s$
						\State Compute $UB(\overline Y)$ and $LB(\overline Y)$
						\If {$UB(\overline Y) + i + 1 < CUB$}
						\State $CUB = UB(\overline Y) + i + 1$
						\EndIf
						\State Add $\overline Y$ to $Q$						
						\EndIf
					\EndFor
				\EndProcedure
				\Procedure{Best nodes selection}{}
					\State Order the nodes $Y$ of $Q$ according to the value $UB(Y)$ (use $LB(Y)$ to break ties)
					\State Let $S_{i+1}$ be the set of the best $\beta$ nodes of $Q$
					\State Remove from $S_{i+1}$ all nodes $Y$ with  $LB(Y) + i \geq CUB$
				\EndProcedure
			\EndIf	
		\EndFor
		\State $i = i + 1$
	\EndWhile
\end{algorithmic}

The algorithm proceeds by iteratively constructing the levels of the enumeration tree. Each level $S_i$
contains the nodes (yards) that can be obtained from the starting yard ($Y_0$) by performing
exactly $i$ reshuffle operations. The value of the current best solution is stored in the variable $CUB$.
The algorithm terminates when it constructs an empty level. 

At each generic iteration $i$, we consider all the nodes $Y \in S_i$ and proceed as follows:
we first remove from $Y$ all the blocks that can be retrieved without requiring a reshuffle 
operation (procedure RETRIEVAL).
Then, let $u$ be the incumbent block of $Y$ that has to be reshuffled, construct all
yards $\overline Y$, each associated with the movements of $u$ to feasible (not full) 
stacks of the yard, and compute $UB(\overline Y)$ and $LB(\overline Y)$.
Observe that $\overline Y$ has been obtained from the starting yard by performing $i + 1$ reshuffles.
Therefore, $UB(\overline Y) + i + 1$ represents an upper bound on the value of the optimal solution
of the problem. Then, in case, update the current $CUB$ value.  
All the nodes $\overline Y$ obtained in this way define the set $Q$ of candidate descendants.
Now, in order to select the nodes of $Q$ that will be used to construct the next level $S_{i+1}$, we
first order them according to the values of the upper bound $UB$ (in case of ties, we use $LB$ to define a priority).
Then, we select the first $\beta$ nodes according to this order.
Finally, we remove from this set the nodes $Y$ with $LB(Y) + i \geq CUB$, since they cannot generate a solution with a better value than $CUB$.

The performances of the Bounded Beam Search heuristic strongly depend on the choice of the algorithms used to calculate
$UB(Y)$ and $LB(Y)$, as well as the value of parameter $\beta$.
Indeed, an accurate procedure for $UB(Y)$ determines a more effective selection of the nodes
that are generated and, then, hopefully, a better final solution.
On the other hand, a good value of $LB(Y)$ helps to reduce the solution space and, therefore,
to decrease the total computing time. 
Clearly, it is usually the case that more accurate algorithms for the evaluation of $UB(Y)$ and
$LB(Y)$ require more computational time. 
Moreover, an higher value for $\beta$ enlarges the search space (and then the probability to find a good solution)
but also increases the total requested time. 
Then it is important to consider all such aspects to obtain good performances that balance quality and time. 
In Section \ref{helb}, we discuss how to tune these parameters in an efficient way. 

\begin{table}[htbp]
\begin{center}
\begin{tabular}{|c|c|c|c|c|c|}
\hline
Name &  Authors \\ \hline
{\sc group 1}  & \citet{WLT2009}\\ \hline 
{\sc group 2} & \citet{WT2010} \\ \hline 
{\sc group 3} & \citet{LL2010} \\ \hline
{\sc group 4} & \citet{CVS2011} \\ \hline
{\sc group 5} & \citet{UA2012} \\ \hline 
{\sc group 6} & \citet{ZQLZ2012} \\ \hline 
\end{tabular}
\end{center}
\caption{Test beds used in the literature of the BRP}
\label{inst_clas}
\end{table}

\section{Description of the Data Set}
\label{instances}

We considered several sets of instances presented in the literature, divided into
six groups that are resumed in Table~\ref{inst_clas}.
In the following, we first discuss datasets that are already available. Then we
introduce the new set of instances {\sc lbri}.
The set {\sc lbri} and the instances of {\sc group 1},..., {\sc group 6} are available at~\citet{B17}.

{\sc group} 1 is defined by the instances presented in \citet{WLT2009}, all with six
stacks of height between two and five. The storage densities of the stacks are
classified as light (L), medium (M) and heavy (H), according to the percentage
of operating capacity taken up. In particular, an instance with $w$ stacks of
height $h$ is light if the number $n$ of blocks is
$0.2 \times ((w-1) \times h + 1)$, is medium if
$n = 0.5 \times ((w-1) \times h + 1)$, and heavy if
$n = 0.8 \times ((w-1) \times h + 1)$.
In total, 12 combinations of stack sizes and storage densities are considered,
each composed by 50 instances. 

The instances of {\sc group} 2 (\citet{WT2010}) are grouped according to the number of
stacks $w \in [3,12]$ and height $h \in [3,12]$.
Therefore, we have 100 problem classes, each with 40 instances.
For each instance, the number of blocks $n = w \times h - (h - 1)$.

The instances of {\sc group} 3 (\citet{LL2010}) are defined on multiple bays yards.
Here, we consider only the subset of the 70 {\em Individual} ones,
i.e., those where the blocks have distinct priorities.
The instances can be {\em Random} (R), when the blocks are randomly located in the
available slots, or {\em Upside down} (U), when blocks with higher retrieval
priority are located in the lower slots of the stacks. All of them have
up to 10 bays, $w \in [3,16]$, $h \in \{6,8\}$, and $n \in [70,720]$.  
Here we allow reshuffles of blocks between different bays:
therefore, an instance with 10 bays and 16 stacks is equivalent to an instance
with 1 bay and 160 stacks.

{\sc group} 4 includes a set of 880 instances introduced in \citet{CVS2011}.
Such instances are randomly generated
with $w \in \{3,4,5,6,10,100\}$, $h \in \{5,6,...,12,102\}$ and
$n = w \times (h - 2)$. 

The 8000 instances of {\sc group} 5 (\citet{UA2012}) are generated with
$w \in [3,7]$, $h \in [4,7]$ and different densities of the blocks.
They can be {\em balanced} (b), if the stacks have the same number of
stored blocks, or  {\em unbalanced} (u). 

In {\sc group} 6, we consider the 12500 instances described in \citet{ZQLZ2012}.
Such instances are grouped into 125 classes of size 100,
according to the values of $w \in [6,10]$, $h \in [3,7]$,
and $n \in [(w - 1) \times h, w \times h - 1]$.

We now introduce a new set of instances, the Large Block Relocation Instances ({\sc lbri}),
that we generated in order to evaluate the performances of the Bounded Beam Search algorithm on a realistic-sized
test bed.
Indeed, except for the ones of {\sc group} 3 and {\sc group} 4,
the largest yards considered in the literature have a size of at most
12 stacks of height 12 and less than 140 blocks.
This is far away from a realistic scenario, where a yard is usually defined by
hundreds of stacks with thousands of blocks. 
The procedure implemented to generate the new instances follows the scheme
introduced in \citet{ZQLZ2012}. Here the authors show that it is always
possible to retrieve all the blocks of a yard with $w$ stacks of height $h$
and $n$ blocks if any block $i$ is located in a slot $t_i \leq h$
such that $h -t_i \le w \times h - n + (i - 1)$.
Then, for any triple $(w, h, n)$, with $n \leq w \times h$,
the instance generation procedure is the
following: we first randomly order the $n$ blocks and then, according to this
order, we iteratively assign each block to the first available slot of an
available randomly selected stack. If the yard generated this way does not fit
the feasibility condition defined above, then it is discarded and regenerated.
The set {\sc lbri} consists of 8400 instances, 100 for each triple $w \in \{50,100,500,1000\}$, $h \in \{4,7,10\}$
and $n \in [w \times (h - 1), w \times h -1]$.

\section{Computing the Unordered Blocks Assignment
Lower Bound} \label{sec:reslb}

Here we test the proposed Unordered Blocks Assignment
Lower Bound and compare the results with the ones obtained by the other lower bounds in the literature, namely LB1, LB3, LB4.
Such comparison is reported in Table \ref{lbn}.
Each row of the table represents a dataset.
Columns LB1, LB3, LB4, and UBALB report average lower bounds (column {\em Value}) and computing times (column {\em Time}) on the instances of the considered
dataset.
The symbol $\ast$ is used to indicated that some instances of the corresponding
group could not be solved within the time limit of one hour by the considered algorithm.

The computational results emphasize the theoretical dominance rules
among LB1, LB3, LB4, and the Unordered Blocks Assignment Lower Bound that
we mentioned in Section~\ref{sec:newlowerbound}.
They also show
that, in practice, LB4 defines better bounds than UBALB.
However, the computational times required
to calculate $G^T(M, B, \phi)$ are, in particular on the large instances, much higher than
the ones needed for solving GMBIP$^B$. Indeed, recall that the algorithm
proposed by \citet{TT2016} is exponential in the number of reshuffled
blocks, while Algorithm 1 runs in polynomial time. Therefore, even if LB4 produces better bounds,
it can hardly be used for tackling real size BRP instances.
On the other hand, the Unordered Blocks Assignment
Lower Bound seems to present a good compromise between
quality of the solutions and computational times, as it always outperforms
the values provided by LB1 and LB3 in, essentially, the same computing time.

\begin{table}[]
\centering
{\fontsize{6}{11}\selectfont
\begin{tabular}{l|l|l|l|r|cc|cc|cc|cc|}
\cline{6-13}
\multicolumn{1}{l}{}                            & \multicolumn{1}{l}{}              & \multicolumn{1}{l}{}                  & \multicolumn{1}{l}{}            &       & \multicolumn{2}{c|}{LB1}                               & \multicolumn{2}{c|}{LB3}                               & \multicolumn{2}{c|}{LB4}                               & \multicolumn{2}{c|}{UBALB}                               \\ \hline
\multicolumn{1}{|c|}{{\sc set}}                     & \multicolumn{1}{c|}{$n$}            & \multicolumn{1}{c|}{$w$}                & \multicolumn{1}{c|}{$h$}          & \multicolumn{1}{c|}{{\scriptsize \#}{\sc i}}       & \multicolumn{1}{l|}{Value} & \multicolumn{1}{l|}{Time} & \multicolumn{1}{l|}{Value} & \multicolumn{1}{l|}{Time} & \multicolumn{1}{l|}{Value} & \multicolumn{1}{l|}{Time} & \multicolumn{1}{l|}{Value} & \multicolumn{1}{l|}{Time} \\ \hline
\multicolumn{1}{|c|}{\multirow{1}{*}{{\sc group 1}}}       & $[3,21]$   & 6              & $[2,5]$  & 600   & 2.93                       & 0.00                         & 3.16                       & 0.00                        & 3.20                        & 0.00                        & 3.16                       & 0.00                        \\ \hline
\multicolumn{1}{|c|}{\multirow{1}{*}{{\sc group 2}}}        & $[7,133]$    & $[3,12]$      & $[3,12]$ & 4000  & 32.71                      & 0.00                        & 39.01                      & 0.00                        & 41.13                      & 0.00                        & 39.11                      & 0.00                        \\ \hline
\multicolumn{1}{|c|}{\multirow{1}{*}{{\sc group 3}}}          & $[70,720]$     & $[16,160]$         & $[6,8]$  & 14    & 213.11                      & 0.00                        & 215.29                     & 0.00                        & 215.61                      & 0.00                        & 215.33                      & 0.00                        \\ \hline
\multicolumn{1}{|c|}{\multirow{2}{*}{{\sc group 4}}}   & $[9,100]$    & $[3,10]$      & $[5,12]$ & 840   & 18.37                      & 0.00                        & 22.04                      & 0.00                        & 22.87                      & 0.00                        & 22.05                      & 0.00                        \\ \cdashline{2-13}
\multicolumn{1}{|c|}{}                          & $10000$          & $100$ & $102$        & 40    & 9485.92                    & 0.00                     & 11177.08                   & 0.00                     & $\ast$                          & $\ast$                         & 11177.52                   & 0.00                     \\ \hline
\multicolumn{1}{|c|}{\multirow{1}{*}{{\sc group 5}}} & $[3,7]$    & $[6,36]$       & $[4,7]$  & 8000  & 8.00                          & 0.00                        & 9.53                       & 0.00                        & 9.81                       & 0.00                        & 9.54                       & 0.00                        \\ \hline
\multicolumn{1}{|c|}{\multirow{1}{*}{{\sc group 6}}}      & $[15,69]$       & $[6,10]$   & $[3,7]$ & 12500 & 21.46                      & 0.00                        & 25.87                      & 0.00                        & 26.59                      & 0.00                        & 25.93                      & 0.00                         \\ \hline
\multicolumn{1}{|c|}{\multirow{1}{*}{{\sc lbri}}}      & $[50,1000]$       & $[4,10]$   & $[199,9999]$ & 8400 & 2140.27                      & 0.00                        & 2156.50                      & 0.00                        & 2160.93                      & 0.03                        & 2156.65                      & 0.00                         \\ \hline
\end{tabular}}
\vspace{.3cm}
\caption{Comparative analysis on different lower bounds obtained on the seven datasets described in Section \ref{instances}}
\label{lbn}
\end{table}

\section{Implementation of the Bounded Beam Search algorithm}
\label{helb}

In this section, we give details on the implementation of the Bounded Beam Search algorithm.
In particular, as a consequence of the computational results presented in Section \ref{sec:reslb}, we use the Unordered Blocks Assignment
Lower Bound as lower bound. 
The choice of $\beta$ and the selection of the algorithms used to compute the value $UB$ are
addressed in Sections \ref{beta} and \ref{ub} respectively.
For all the tests and the final computational experience,
we use a machine with a CPU Intel Core i7-3632QM at 2.2GHz and with 1 GB RAM under a Linux
operating system. We set one second as time limit for the solution of each instance.

\subsection{The parameter $\beta$}
\label{beta}

As discussed in Section \ref{sec:beamserch}, choosing the right value for $\beta$ is crucial for the performances
of the Bounded Beam Search algorithm, as it controls the size of the search space. Indeed,
a larger $\beta$ implies a larger
number of descendants of each node of the search tree and, therefore, of solutions explored.
On the other hand, this also corresponds to larger computing times.
Notice that, when $\beta = 1$, one obtains a greedy algorithm,
while an exact approach corresponds to an unrestricted $\beta$.

We test values $\{50,$ $100,$ $200,$ $300,$ $400,$ $500,$ $600,$ $700,$ $800\}$ for $\beta$.
We use ChainF to compute $UB$, as it seems the best compromise between computational time and quality, according to the results in the literature (\citet{JV2014}).
The computational results are shown in Table \ref{bsbeta}.
We grouped the instances in the literature into classes according to the size (number of blocks $n$).
Each column of the table represents a given class and each row is associated with a value of $\beta$.
Then, in each entry $(\bar \beta, \bar r)$ of the table,
we report the number of times the algorithm with $\beta = \bar \beta$ obtains the best results on the instances
with $n \in \bar r$.  

For each column, the best value is in bold. As one could expect, as soon as the instance size increases,
the best performances are obtained for smaller values of $\beta$.
Indeed, without a time limit, the best performance would be obtained for an unrestricted $\beta$
(exact approach).
Whereas, if a time limit is given, things are different.
In fact, if we generate a large number of yards (large values of $\beta$), we increase the choice of finding good solutions.
However, we may not be able, due to the time limit, to process the generated yards and, hence, to find those solutions.
On the other hand, if we generate a small number of yards (small values of $\beta$), we increase the choice of being able to process all of them within the time limit.
At the same time, we may miss (not generate) yards corresponding to good solutions.
Therefore, a good value of $\beta$ for an algorithm with a time limit, is a compromise between the number of yards generated and the time we have to process them. 

Therefore, we implement the Bounded Beam Search procedure by setting $\beta$
according to the size of the input instance.
In particular: $\beta = 800$, if $n \in [0, 40)$; $\beta = 500$, if $n \in [40, 60)$; $\beta = 300$,
    if $n \in [60, 80)$; $\beta = 200$, if $n \in [80, 100)$; $\beta = 100$, if $n \in [100, 120)$; $\beta = 50$,
          if $n \ge 120$.

\begin{table}
\centering
{\fontsize{7}{10}\selectfont
\setlength\tabcolsep{2.5pt}
\begin{tabular}{cc|c|c|c|c|c|c|c|c|}
\cline{2-10}
\multicolumn{1}{c|}{} & \multicolumn{1}{c|}{$n < 40$} & \multicolumn{1}{c|}{$[40,60)$} & \multicolumn{1}{c|}{$[60,80)$} & \multicolumn{1}{c|}{$[80,100)$} & \multicolumn{1}{c|}{$[100,120)$} & \multicolumn{1}{c|}{$[120,140)$} & \multicolumn{1}{c|}{$[140,200)$} & \multicolumn{1}{c|}{$[200,400)$}  & \multicolumn{1}{c|}{$[400,500)$}   \\ \hline
\multicolumn{1}{|c|}{50} & 17263 & 5169 & 1057 & 156 & 139 & \textbf{81} & 397 & \textbf{1043} & \textbf{919} \\ \hline 
\multicolumn{1}{|c|}{100} & 17361 & 5439 & 1241 & 220 & \textbf{195} & 63 & \textbf{398} & 614 & 238 \\ \hline 
\multicolumn{1}{|c|}{200} & 17423 & 5655 & 1383 & \textbf{268} & 98 & 6 & \textbf{398} & 475 & 117 \\ \hline 
\multicolumn{1}{|c|}{300} & 17448 & 5779 & \textbf{1489} & 192 & 38 & 5 & 390 & 429 & 89 \\ \hline 
\multicolumn{1}{|c|}{400} & 17459 & 5851 & 1460 & 124 & 5 & 0 & 388 & 398 & 77 \\ \hline 
\multicolumn{1}{|c|}{500} & 17473 & \textbf{5888} & 1350 & 71 & 2 & 1 & 384 & 382 & 72 \\ \hline 
\multicolumn{1}{|c|}{600} & 17486 & 5881 & 1238 & 55 & 0 & 0 & 382 & 376 & 66 \\ \hline 
\multicolumn{1}{|c|}{700} & \textbf{17496} & 5817 & 1110 & 32 & 0 & 0 & 380 & 371 & 63 \\ \hline 
\multicolumn{1}{|c|}{800} & 17495 & 5737 & 1001 & 22 & 0 & 0 & 377 & 367 & 60 \\ \hline 
\\ 
\cline{2-10}
\multicolumn{1}{c|}{} & \multicolumn{1}{c|}{$[500,600)$} & \multicolumn{1}{c|}{$[600,700)$} & \multicolumn{1}{c|}{$[700,1000)$}  & \multicolumn{1}{c|}{$[1000,2000)$} & \multicolumn{1}{c|}{$[2000,4000)$} & \multicolumn{1}{c|}{$[4000,5000)$} & \multicolumn{1}{c|}{$[5000,7000)$} & \multicolumn{1}{c|}{$[7000,10000)$} & \multicolumn{1}{c|}{$n \geq 10000$}  \\ \hline
\multicolumn{1}{|c|}{50} & \textbf{14} & \textbf{665} & \textbf{930} & \textbf{399} & \textbf{1100} & \textbf{1000} & \textbf{700} & \textbf{1000} & \textbf{40} \\ \hline 
\multicolumn{1}{|c|}{100} & \textbf{14} & 295 & 524 & 394 & 1095 & 986 & 698 & 975 & 39 \\ \hline 
\multicolumn{1}{|c|}{200} & \textbf{14} & 199 & 396 & 394 & 1082 & 968 & 684 & 935 & 38 \\ \hline 
\multicolumn{1}{|c|}{300} & \textbf{14} & 171 & 371 & 393 & 1082 & 952 & 678 & 896 & 34 \\ \hline 
\multicolumn{1}{|c|}{400} & \textbf{14} & 158 & 359 & 393 & 1076 & 928 & 672 & 862 & 30 \\ \hline 
\multicolumn{1}{|c|}{500} & \textbf{14} & 155 & 358 & 393 & 1066 & 908 & 658 & 828 & 27 \\ \hline 
\multicolumn{1}{|c|}{600} & \textbf{14} & 152 & 354 & 392 & 1057 & 898 & 648 & 776 & 25 \\ \hline 
\multicolumn{1}{|c|}{700} & \textbf{14} & 152 & 352 & 392 & 1053 & 881 & 628 & 736 & 20 \\ \hline 
\multicolumn{1}{|c|}{800} & \textbf{14} & 152 & 350 & 392 & 1043 & 856 & 613 & 690 & 17 \\ \hline 
\end{tabular}
}
\caption{Comparative analysis for different values of $\beta$}
\label{bsbeta}
\end{table}

\subsection{The algorithm for $UB$}
\label{ub}

The performances of a beam search procedure heavily depend on the algorithm used to compute an upper bound
$UB$ for each node of the search tree, in order to select its most promising descendants.
To choose the one that fits better in the Bounded Beam Search heuristic, we test algorithms Expected Reshuffle Index, Difference1, Probe Restricted 4, Chain, Group Assignment Heuristic, ChainF, the heuristic procedure proposed in \citet{CSV2012},
that are the best fast heuristics in the literature.
In the experiments, we set $\beta$ as described in Section \ref{beta}.

Table \ref{resultsheuristic} reports the computational results obtained.
Each row corresponds to one of the $UB$ algorithm taken under consideration (ERI denotes the Expected Reshuffle Index heuristic, PR4 the Probe Restricted 4, GAH the Group Assignment Heuristic, Min-Max the heuristic procedure proposed in \citet{CSV2012}), while, as in Table \ref{bsbeta}, each column represents a class of instances, defined according to the number of blocks $n$.

For each group of instances and $UB$ procedure, we report the number of times the corresponding
algorithm obtains the best performances.
According to the results in the table, we implement the Bounded Beam Search
procedure with ChainF when $n < 1000$,
with Difference1 for $n \in [1000,10000)$ and with GAH when $n \ge 10000$.

\begin{table}
\centering
{\fontsize{7}{10}\selectfont
\setlength\tabcolsep{2.5pt}
\begin{tabular}{cc|c|c|c|c|c|c|c|c|}
\cline{2-10}
\multicolumn{1}{c|}{} & \multicolumn{1}{c|}{$n<40$} & \multicolumn{1}{c|}{$[40,60)$} & \multicolumn{1}{c|}{$[60,80)$} & \multicolumn{1}{c|}{$[80,100)$} & \multicolumn{1}{c|}{$[100,120)$} & \multicolumn{1}{c|}{$[120,140)$} & \multicolumn{1}{c|}{$[140,200)$} & \multicolumn{1}{c|}{$[200,400)$}  & \multicolumn{1}{c|}{$[400,500)$}   \\ \hline
\multicolumn{1}{|c|}{Min-Max} & 17420 & 5456 & 1074 & 81 & 42 & 25 & 393 & 436 & 14 \\ \hline 
\multicolumn{1}{|c|}{ERI} & 16939 & 3458 & 192 & 5 & 0 & 6 & 279 & 202 & 6 \\ \hline 
\multicolumn{1}{|c|}{Difference1} & 17402 & 5303 & 988 & 73 & 20 & 11 & 393 & 575 & 51 \\ \hline 
\multicolumn{1}{|c|}{PR4} & 17429 & 5488 & 1092 & 107 & 46 & 20 & 396 & 561 & 88 \\ \hline 
\multicolumn{1}{|c|}{Chain} & 17452 & 5616 & 1176 & 127 & 54 & 24 & 392 & 447 & 23 \\ \hline 
\multicolumn{1}{|c|}{GAH} & \textbf{17455} & 5662 & 1238 & 148 & \textbf{113} & \textbf{54} & 392 & 451 & 155 \\ \hline 
\multicolumn{1}{|c|}{ChainF} & 17451 & \textbf{5681} & \textbf{1275} & \textbf{156} & 97 & 45 & \textbf{399} & \textbf{682} & \textbf{766} \\ \hline 
\\ 
\cline{2-10}
\multicolumn{1}{c|}{} & \multicolumn{1}{c|}{$[500,600)$} & \multicolumn{1}{c|}{$[600,700)$} & \multicolumn{1}{c|}{$[700,1000)$}  & \multicolumn{1}{c|}{$[1000,2000)$} & \multicolumn{1}{c|}{$[2000,4000)$} & \multicolumn{1}{c|}{$[4000,5000)$} & \multicolumn{1}{c|}{$[5000,7000)$} & \multicolumn{1}{c|}{$[7000,10000)$} & \multicolumn{1}{c|}{$n  \geq 10000$}  \\ \hline
\multicolumn{1}{|c|}{Min-Max} & 10 & 10 & 0 & 253 & 246 & 0 & 0 & 0 & 0 \\ \hline 
\multicolumn{1}{|c|}{ERI} & 8 & 9 & 0 & 161 & 148 & 0 & 0 & 0 & 0 \\ \hline 
\multicolumn{1}{|c|}{Difference1} & \textbf{13} & 282 & 105 & \textbf{298} & \textbf{935} & \textbf{958} & \textbf{689} & \textbf{1000} & 0 \\ \hline 
\multicolumn{1}{|c|}{PR4} & 11 & 46 & 10 & 275 & 280 & 0 & 0 & 0 & 0 \\ \hline 
\multicolumn{1}{|c|}{Chain} & 10 & 16 & 0 & 254 & 245 & 0 & 0 & 0 & 0 \\ \hline 
\multicolumn{1}{|c|}{GAH} & 11 & 42 & 64 & 225 & 223 & 0 & 0 & 0 & \textbf{40} \\ \hline 
\multicolumn{1}{|c|}{ChainF} & 12 & \textbf{399} & \textbf{838} & 294 & 380 & 42 & 16 & 0 & 0 \\ \hline 
\end{tabular}
}
\caption{Comparative analysis for different $UB$ algorithms}
\label{resultsheuristic}
\end{table}

\section{Computational results}
\label{sec:results}

In this Section we compare the performances of the Bounded Beam Search
algorithm with the ones of the other approaches in the
literature on the instances of {\sc group 1},$\dots$, {\sc group 6} and {\sc lbri}.
The algorithms we consider are taken from the ones introduced in Section \ref{intro}.
We reimplemented the procedures Difference1, Chain, ChainF, Expected Reshuffle Index, Probe Restricted 4, the heuristic proposed in \citet{CSV2012}, Group Assignment Heuristic.
The computational results for the other algorithms are taken from the literature: see \citet{JV2014},
\citet{WT2010}, \citet{ZQLZ2012}, \citet{CVS2011}, \citet{CSV2009}, \citet{LL2010}, \citet{KA2016b}, \citet{WLT2009}, and
\citet{UA2012}.
Moreover, we also consider the exact procedure of \citet{TT2016} (whose code was provided us by the authors)
setting a time limit of one second. We denote by BBT this heuristic variant of the method. Moreover,  Bounded Beam Search procedure will be addressed as BBS. 
A time limit of one second is also set for the Bounded Beam Search algorithm
as well as for the other procedures we reimplemented.
However, the results taken from the literature may not respect this limit and BBT may slightly exceed it.

For each set of instances, we report a table with detailed results.
In the tables, $w$ indicates the number of stacks, $h$ the number of available slots for each stack, $n$ the number
of blocks, and {\scriptsize \#}{\sc i} is the number of instances of the group.
Resh and Time denote the average number of reshuffles and the average computing time, respectively.
Resh is in bold for the (possibly not unique) best approach.
In order to evaluate the quality of the solutions provided by the heuristic algorithms, we also present, in column Opt,
the average number of reshuffles obtained by the exact procedure of \citet{TT2016}
within a time limit of 1800 seconds. Symbol ``$\ast$'' in the columns indicates that the corresponding procedure runs out
of time on at least one instance of the group.
Symbol``$\blacktriangle$'' denotes that Resh = Opt.
For the instances of {\sc group 1},$\dots$, {\sc group 6}, the tables report the solutions obtained by BBS, BBT
and, in column BL, the best results provided by the algorithms that we reimplemented (Difference1, Chain, ChainF, Expected Reshuffle Index, Probe Restricted 4, the heuristic proposed in \citet{CSV2012}, Group Assignment Heuristic).
Moreover, each table may include additional columns corresponding to algorithms we did not re-implement.
For each algorithm we also report the machine used for the experiments.
Observe that the machine we use never outperforms the ones used for the algorithms in the additional columns.
For every table we also present a figure that summarizes the detailed results.
Such figures report a column for each algorithm considered in the corresponding table.
The height of each column measures the number of times ($\# wins$) the associated algorithm is the best approach.

In Table \ref{tab:g1} and Figure \ref{fig:g1},
we show the computational results obtained on the instances of {\sc group 1}.
Column MRIP$_k$ reports the results output by the Minimum Reshuffle Integer Program heuristic with $k = 6$ which is the best heuristic algorithm presented in \citet{WLT2009}.
Since the instances are quite small, all the approaches are rather effective and
both BBT and the Bounded Beam Search algorithm always provide the optimal
solutions. 

\begin{table}
\begin{center}
{\scriptsize
\begin{tabular}{cccc|cc|cc|cc|cc|c}
\cline{5-12}
 & & & \multicolumn{1}{c|}{} & \multicolumn{2}{c|}{MRIP$_k$$^a$}  & \multicolumn{2}{c|}{BL$^b$} & \multicolumn{2}{c|}{BBT$^b$}  & \multicolumn{2}{c|}{BBS$^b$} & \multicolumn{1}{c}{} \\ \hline
\multicolumn{1}{|c|}{$w$} & \multicolumn{1}{c|}{$h$}  & \multicolumn{1}{c|}{$n$} & \multicolumn{1}{c|}{{\scriptsize \#}{\sc i}}  & \multicolumn{1}{c|}{Resh} & \multicolumn{1}{c|}{Time} & \multicolumn{1}{c|}{Resh} & \multicolumn{1}{c|}{Time} & \multicolumn{1}{c|}{Resh} & \multicolumn{1}{c|}{Time} & \multicolumn{1}{c|}{Resh} & \multicolumn{1}{c|}{Time} & \multicolumn{1}{c|}{Opt}   \\ \hline
\multicolumn{1}{|c|}{6} & \multicolumn{1}{c|}{2} & \multicolumn{1}{c|}{9} & \multicolumn{1}{c|}{50} & \textbf{1.70}$^\blacktriangle$ & 0.181  & \textbf{1.70}$^\blacktriangle$ & 0.000 & \textbf{1.70}$^\blacktriangle$ & 0.002 & \textbf{1.70}$^\blacktriangle$ & 0.002 & \multicolumn{1}{c|}{1.70} \\ \cline{1-4} 
\multicolumn{1}{|c|}{6} & \multicolumn{1}{c|}{3} & \multicolumn{1}{c|}{13} & \multicolumn{1}{c|}{50} & \textbf{4.58}$^\blacktriangle$ & 0.877  & \textbf{4.58}$^\blacktriangle$ & 0.000 & \textbf{4.58}$^\blacktriangle$ & 0.002 & \textbf{4.58}$^\blacktriangle$ & 0.002 & \multicolumn{1}{c|}{4.58} \\ \cline{1-4} 
\multicolumn{1}{|c|}{6} & \multicolumn{1}{c|}{4} & \multicolumn{1}{c|}{17} & \multicolumn{1}{c|}{50} & \textbf{7.56}$^\blacktriangle$ & 6.437  & 7.60 & 0.000 & \textbf{7.56}$^\blacktriangle$ & 0.002 & \textbf{7.56}$^\blacktriangle$ & 0.003 & \multicolumn{1}{c|}{7.56} \\ \cline{1-4} 
\multicolumn{1}{|c|}{6} & \multicolumn{1}{c|}{5} & \multicolumn{1}{c|}{21} & \multicolumn{1}{c|}{50} & 11.80 & 66.809  & 11.84 & 0.000 & \textbf{11.68}$^\blacktriangle$ & 0.002 & \textbf{11.68}$^\blacktriangle$ & 0.004 & \multicolumn{1}{c|}{11.68} \\ \hline 
\end{tabular}
}
\end{center}
{\scriptsize $^a$Intel dual core Xeon 3GHz and 4 GB RAM}\\
{\scriptsize $^b$Intel Core i7-3632QM 2.2GHz and 1 GB RAM}
\caption{Computational results on the set {\sc group 1}}
\label{tab:g1}
\end{table}

\begin{figure}[htbp] 
  \begin{center}

\includegraphics[width=6cm,height=5cm]{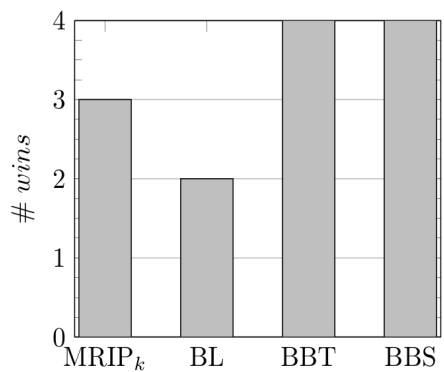}

    \end{center}
\caption{Aggregated results on the set {\sc group 1}}
\label{fig:g1}
\end{figure}

Table \ref{tab:g2a}, Table \ref{tab:g2b} and Figure \ref{fig:g2} report the computational results on the
instances of {\sc group 2}.
Here we did not consider the results in \citet{WT2010} since, as confirmed us by the authors, they contain some errors.
With only two exceptions, the Bounded Beam Search heuristic is always the best approach and, especially for the larger instances, outperforms all
the other algorithms. Moreover, in 41 out of 100 subsets of 40 instances, the solutions provided by the Bounded Beam Search algorithm are certified to be optimal
for any instance.  

\begin{table}
\begin{center}
{\scriptsize
\begin{tabular}{cccc|cc|cc|cc|c}
\cline{5-10}
 & & & \multicolumn{1}{c|}{} & \multicolumn{2}{c|}{BL$^a$} & \multicolumn{2}{c|}{BBT$^a$} & \multicolumn{2}{c|}{BBS$^a$} & \multicolumn{1}{c}{} \\ \hline
\multicolumn{1}{|c|}{$w$} & \multicolumn{1}{c|}{$h$} & \multicolumn{1}{c|}{$n$} & \multicolumn{1}{c|}{{\scriptsize \#}{\sc i}}  & \multicolumn{1}{c|}{Resh} & \multicolumn{1}{c|}{Time}  & \multicolumn{1}{c|}{Resh} & \multicolumn{1}{c|}{Time} & \multicolumn{1}{c|}{Resh} & \multicolumn{1}{c|}{Time} & \multicolumn{1}{c|}{Opt}   \\ \hline
\multicolumn{1}{|c|}{3} & \multicolumn{1}{c|}{3} & \multicolumn{1}{c|}{7} & \multicolumn{1}{c|}{40} & \textbf{3.30}$^\blacktriangle$ & 0.000 & \textbf{3.30}$^\blacktriangle$ & 0.002 & \textbf{3.30}$^\blacktriangle$ & 0.002 & \multicolumn{1}{c|}{3.30} \\ \cline{1-4} 
\multicolumn{1}{|c|}{3} & \multicolumn{1}{c|}{4} & \multicolumn{1}{c|}{9} & \multicolumn{1}{c|}{40} & 5.70 & 0.000 & \textbf{5.68}$^\blacktriangle$ & 0.002 & \textbf{5.68}$^\blacktriangle$ & 0.002 & \multicolumn{1}{c|}{5.68} \\ \cline{1-4} 
\multicolumn{1}{|c|}{3} & \multicolumn{1}{c|}{5} & \multicolumn{1}{c|}{11} & \multicolumn{1}{c|}{40} & 8.50 & 0.000 & \textbf{8.40}$^\blacktriangle$ & 0.002 & \textbf{8.40}$^\blacktriangle$ & 0.002 & \multicolumn{1}{c|}{8.40} \\ \cline{1-4} 
\multicolumn{1}{|c|}{3} & \multicolumn{1}{c|}{6} & \multicolumn{1}{c|}{13} & \multicolumn{1}{c|}{40}  & 12.15 & 0.000 & \textbf{11.50}$^\blacktriangle$ & 0.002 & \textbf{11.50}$^\blacktriangle$ & 0.002 & \multicolumn{1}{c|}{11.50} \\ \cline{1-4} 
\multicolumn{1}{|c|}{3} & \multicolumn{1}{c|}{7} & \multicolumn{1}{c|}{15} & \multicolumn{1}{c|}{40} & 15.55 & 0.000 & \textbf{15.03}$^\blacktriangle$ & 0.002 & \textbf{15.03}$^\blacktriangle$ & 0.003 & \multicolumn{1}{c|}{15.03} \\ \cline{1-4} 
\multicolumn{1}{|c|}{3} & \multicolumn{1}{c|}{8} & \multicolumn{1}{c|}{17} & \multicolumn{1}{c|}{40} & 19.83 & 0.000 & \textbf{18.63}$^\blacktriangle$ & 0.003 & \textbf{18.63}$^\blacktriangle$ & 0.004 & \multicolumn{1}{c|}{18.63} \\ \cline{1-4} 
\multicolumn{1}{|c|}{3} & \multicolumn{1}{c|}{9} & \multicolumn{1}{c|}{19} & \multicolumn{1}{c|}{40} & 25.70 & 0.000 & \textbf{23.98}$^\blacktriangle$ & 0.003 & \textbf{23.98}$^\blacktriangle$ & 0.013 & \multicolumn{1}{c|}{23.98} \\ \cline{1-4} 
\multicolumn{1}{|c|}{3} & \multicolumn{1}{c|}{10} & \multicolumn{1}{c|}{21} & \multicolumn{1}{c|}{40} & 31.35 & 0.000 & \textbf{28.20}$^\blacktriangle$ & 0.006 & \textbf{28.20}$^\blacktriangle$ & 0.042 & \multicolumn{1}{c|}{28.20} \\ \cline{1-4} 
\multicolumn{1}{|c|}{3} & \multicolumn{1}{c|}{11} & \multicolumn{1}{c|}{23} & \multicolumn{1}{c|}{40} & 37.93 & 0.000 & \textbf{32.70}$^\blacktriangle$ & 0.011 & 32.75 & 0.070 & \multicolumn{1}{c|}{32.70} \\ \cline{1-4} 
\multicolumn{1}{|c|}{3} & \multicolumn{1}{c|}{12} & \multicolumn{1}{c|}{25} & \multicolumn{1}{c|}{40} & 42.38 & 0.000 & \textbf{36.43}$^\blacktriangle$ & 0.035 & \textbf{36.43}$^\blacktriangle$ & 0.090 & \multicolumn{1}{c|}{36.43} \\ \cline{1-4} 
\multicolumn{1}{|c|}{4} & \multicolumn{1}{c|}{3} & \multicolumn{1}{c|}{10} & \multicolumn{1}{c|}{40} & \textbf{4.85}$^\blacktriangle$ & 0.000 & \textbf{4.85}$^\blacktriangle$ & 0.002 & \textbf{4.85}$^\blacktriangle$ & 0.002 & \multicolumn{1}{c|}{4.85} \\ \cline{1-4} 
\multicolumn{1}{|c|}{4} & \multicolumn{1}{c|}{4} & \multicolumn{1}{c|}{13} & \multicolumn{1}{c|}{40} & 8.53 & 0.000 & \textbf{8.43}$^\blacktriangle$ & 0.002 & \textbf{8.43}$^\blacktriangle$ & 0.002 & \multicolumn{1}{c|}{8.43} \\ \cline{1-4} 
\multicolumn{1}{|c|}{4} & \multicolumn{1}{c|}{5} & \multicolumn{1}{c|}{16} & \multicolumn{1}{c|}{40} & 12.55 & 0.000 & \textbf{12.25}$^\blacktriangle$ & 0.002 & \textbf{12.25}$^\blacktriangle$ & 0.003 & \multicolumn{1}{c|}{12.25} \\ \cline{1-4} 
\multicolumn{1}{|c|}{4} & \multicolumn{1}{c|}{6} & \multicolumn{1}{c|}{19} & \multicolumn{1}{c|}{40} & 16.35 & 0.000 & \textbf{15.63}$^\blacktriangle$ & 0.002 & \textbf{15.63}$^\blacktriangle$ & 0.004 & \multicolumn{1}{c|}{15.63} \\ \cline{1-4} 
\multicolumn{1}{|c|}{4} & \multicolumn{1}{c|}{7} & \multicolumn{1}{c|}{22} & \multicolumn{1}{c|}{40} & 24.23 & 0.000 & \textbf{22.60}$^\blacktriangle$ & 0.003 & \textbf{22.60}$^\blacktriangle$ & 0.023 & \multicolumn{1}{c|}{22.60} \\ \cline{1-4} 
\multicolumn{1}{|c|}{4} & \multicolumn{1}{c|}{8} & \multicolumn{1}{c|}{25} & \multicolumn{1}{c|}{40} & 30.65 & 0.000 & \textbf{27.70}$^\blacktriangle$ & 0.006 & \textbf{27.70}$^\blacktriangle$ & 0.053 & \multicolumn{1}{c|}{27.70} \\ \cline{1-4} 
\multicolumn{1}{|c|}{4} & \multicolumn{1}{c|}{9} & \multicolumn{1}{c|}{28} & \multicolumn{1}{c|}{40} & 36.20 & 0.000 & \textbf{32.43}$^\blacktriangle$ & 0.040 & \textbf{32.43}$^\blacktriangle$ & 0.168 & \multicolumn{1}{c|}{32.43} \\ \cline{1-4} 
\multicolumn{1}{|c|}{4} & \multicolumn{1}{c|}{10} & \multicolumn{1}{c|}{31} & \multicolumn{1}{c|}{40} & 45.33 & 0.000 & 40.35 & 0.284 & \textbf{40.15} & 0.287 & \multicolumn{1}{c|}{40.05} \\ \cline{1-4} 
\multicolumn{1}{|c|}{4} & \multicolumn{1}{c|}{11} & \multicolumn{1}{c|}{34} & \multicolumn{1}{c|}{40} & 53.88 & 0.001 & 46.10 & 0.482 & \textbf{45.73} & 0.383 & \multicolumn{1}{c|}{45.60} \\ \cline{1-4} 
\multicolumn{1}{|c|}{4} & \multicolumn{1}{c|}{12} & \multicolumn{1}{c|}{37} & \multicolumn{1}{c|}{40} & 64.90 & 0.001 & 55.55 & 0.823 & \textbf{54.38} & 0.612 & \multicolumn{1}{c|}{$\ast$} \\ \cline{1-4} 
\multicolumn{1}{|c|}{5} & \multicolumn{1}{c|}{3} & \multicolumn{1}{c|}{13} & \multicolumn{1}{c|}{40} & \textbf{5.75}$^\blacktriangle$ & 0.000 & \textbf{5.75}$^\blacktriangle$ & 0.002 & \textbf{5.75}$^\blacktriangle$ & 0.002 & \multicolumn{1}{c|}{5.75} \\ \cline{1-4} 
\multicolumn{1}{|c|}{5} & \multicolumn{1}{c|}{4} & \multicolumn{1}{c|}{17} & \multicolumn{1}{c|}{40} & 11.10 & 0.000 & \textbf{10.98}$^\blacktriangle$ & 0.002 & \textbf{10.98}$^\blacktriangle$ & 0.003 & \multicolumn{1}{c|}{10.98} \\ \cline{1-4} 
\multicolumn{1}{|c|}{5} & \multicolumn{1}{c|}{5} & \multicolumn{1}{c|}{21} & \multicolumn{1}{c|}{40} & 16.10 & 0.000 & \textbf{15.58}$^\blacktriangle$ & 0.002 & \textbf{15.58}$^\blacktriangle$ & 0.005 & \multicolumn{1}{c|}{15.58} \\ \cline{1-4} 
\multicolumn{1}{|c|}{5} & \multicolumn{1}{c|}{6} & \multicolumn{1}{c|}{25} & \multicolumn{1}{c|}{40} & 22.10 & 0.000 & \textbf{21.05}$^\blacktriangle$ & 0.003 & \textbf{21.05}$^\blacktriangle$ & 0.028 & \multicolumn{1}{c|}{21.05} \\ \cline{1-4} 
\multicolumn{1}{|c|}{5} & \multicolumn{1}{c|}{7} & \multicolumn{1}{c|}{29} & \multicolumn{1}{c|}{40} & 30.00 & 0.000 & \textbf{27.53}$^\blacktriangle$ & 0.019 & 27.55 & 0.101 & \multicolumn{1}{c|}{27.53} \\ \cline{1-4} 
\multicolumn{1}{|c|}{5} & \multicolumn{1}{c|}{8} & \multicolumn{1}{c|}{33} & \multicolumn{1}{c|}{40} & 39.55 & 0.000 & 35.85 & 0.187 & \textbf{35.83} & 0.302 & \multicolumn{1}{c|}{35.80} \\ \cline{1-4} 
\multicolumn{1}{|c|}{5} & \multicolumn{1}{c|}{9} & \multicolumn{1}{c|}{37} & \multicolumn{1}{c|}{40} & 47.18 & 0.000 & 41.48 & 0.293 & \textbf{41.40} & 0.404 & \multicolumn{1}{c|}{41.25} \\ \cline{1-4} 
\multicolumn{1}{|c|}{5} & \multicolumn{1}{c|}{10} & \multicolumn{1}{c|}{41} & \multicolumn{1}{c|}{40} & 57.33 & 0.001 & 51.13 & 0.826 & \textbf{50.18} & 0.422 & \multicolumn{1}{c|}{$\ast$} \\ \cline{1-4} 
\multicolumn{1}{|c|}{5} & \multicolumn{1}{c|}{11} & \multicolumn{1}{c|}{45} & \multicolumn{1}{c|}{40} & 68.35 & 0.001 & 60.50 & 1.079 & \textbf{57.78} & 0.606 & \multicolumn{1}{c|}{$\ast$} \\ \cline{1-4} 
\multicolumn{1}{|c|}{5} & \multicolumn{1}{c|}{12} & \multicolumn{1}{c|}{49} & \multicolumn{1}{c|}{40} & 81.33 & 0.000 & 74.53 & 1.242 & \textbf{68.15} & 0.807 & \multicolumn{1}{c|}{$\ast$} \\ \cline{1-4} 
\multicolumn{1}{|c|}{6} & \multicolumn{1}{c|}{3} & \multicolumn{1}{c|}{16} & \multicolumn{1}{c|}{40} & \textbf{7.65}$^\blacktriangle$ & 0.000 & \textbf{7.65}$^\blacktriangle$ & 0.002 & \textbf{7.65}$^\blacktriangle$ & 0.004 & \multicolumn{1}{c|}{7.65} \\ \cline{1-4} 
\multicolumn{1}{|c|}{6} & \multicolumn{1}{c|}{4} & \multicolumn{1}{c|}{21} & \multicolumn{1}{c|}{40} & 12.20 & 0.000 & \textbf{12.03}$^\blacktriangle$ & 0.002 & \textbf{12.03}$^\blacktriangle$ & 0.003 & \multicolumn{1}{c|}{12.03} \\ \cline{1-4} 
\multicolumn{1}{|c|}{6} & \multicolumn{1}{c|}{5} & \multicolumn{1}{c|}{26} & \multicolumn{1}{c|}{40} & 20.08 & 0.000 & \textbf{19.33}$^\blacktriangle$ & 0.003 & \textbf{19.33}$^\blacktriangle$ & 0.012 & \multicolumn{1}{c|}{19.33} \\ \cline{1-4} 
\multicolumn{1}{|c|}{6} & \multicolumn{1}{c|}{6} & \multicolumn{1}{c|}{31} & \multicolumn{1}{c|}{40} & 27.58 & 0.000 & \textbf{25.98}$^\blacktriangle$ & 0.007 & \textbf{25.98}$^\blacktriangle$ & 0.071 & \multicolumn{1}{c|}{25.98} \\ \cline{1-4} 
\multicolumn{1}{|c|}{6} & \multicolumn{1}{c|}{7} & \multicolumn{1}{c|}{36} & \multicolumn{1}{c|}{40} & 37.90 & 0.000 & \textbf{34.43} & 0.182 & 34.45 & 0.305 & \multicolumn{1}{c|}{34.40} \\ \cline{1-4} 
\multicolumn{1}{|c|}{6} & \multicolumn{1}{c|}{8} & \multicolumn{1}{c|}{41} & \multicolumn{1}{c|}{40} & 47.13 & 0.001 & 42.50 & 0.424 & \textbf{42.25} & 0.295 & \multicolumn{1}{c|}{42.10} \\ \cline{1-4} 
\multicolumn{1}{|c|}{6} & \multicolumn{1}{c|}{9} & \multicolumn{1}{c|}{46} & \multicolumn{1}{c|}{40} & 58.45 & 0.001 & 52.93 & 0.986 & \textbf{51.15} & 0.516 & \multicolumn{1}{c|}{$\ast$} \\ \cline{1-4} 
\multicolumn{1}{|c|}{6} & \multicolumn{1}{c|}{10} & \multicolumn{1}{c|}{51} & \multicolumn{1}{c|}{40} & 70.43 & 0.002 & 64.58 & 1.197 & \textbf{61.50} & 0.793 & \multicolumn{1}{c|}{$\ast$} \\ \cline{1-4} 
\multicolumn{1}{|c|}{6} & \multicolumn{1}{c|}{11} & \multicolumn{1}{c|}{56} & \multicolumn{1}{c|}{40} & 84.18 & 0.003 & 76.43 & 1.273 & \textbf{71.35} & 0.972 & \multicolumn{1}{c|}{$\ast$} \\ \cline{1-4} 
\multicolumn{1}{|c|}{6} & \multicolumn{1}{c|}{12} & \multicolumn{1}{c|}{61} & \multicolumn{1}{c|}{40} & 99.75 & 0.000 & 93.08 & 1.271 & \textbf{82.73} & 0.738 & \multicolumn{1}{c|}{$\ast$} \\ \cline{1-4} 
\multicolumn{1}{|c|}{7} & \multicolumn{1}{c|}{3} & \multicolumn{1}{c|}{19} & \multicolumn{1}{c|}{40} & \textbf{8.95}$^\blacktriangle$ & 0.000 & \textbf{8.95}$^\blacktriangle$ & 0.002 & \textbf{8.95}$^\blacktriangle$ & 0.003 & \multicolumn{1}{c|}{8.95} \\ \cline{1-4} 
\multicolumn{1}{|c|}{7} & \multicolumn{1}{c|}{4} & \multicolumn{1}{c|}{25} & \multicolumn{1}{c|}{40} & 15.68 & 0.000 & \textbf{15.48}$^\blacktriangle$ & 0.002 & \textbf{15.48}$^\blacktriangle$ & 0.004 & \multicolumn{1}{c|}{15.48} \\ \cline{1-4} 
\multicolumn{1}{|c|}{7} & \multicolumn{1}{c|}{5} & \multicolumn{1}{c|}{31} & \multicolumn{1}{c|}{40} & 22.08 & 0.000 & \textbf{21.35}$^\blacktriangle$ & 0.004 & \textbf{21.35}$^\blacktriangle$ & 0.023 & \multicolumn{1}{c|}{21.35} \\ \cline{1-4} 
\multicolumn{1}{|c|}{7} & \multicolumn{1}{c|}{6} & \multicolumn{1}{c|}{37} & \multicolumn{1}{c|}{40} & 32.58 & 0.000 & \textbf{30.75}$^\blacktriangle$ & 0.031 & \textbf{30.75}$^\blacktriangle$ & 0.185 & \multicolumn{1}{c|}{30.75} \\ \cline{1-4} 
\multicolumn{1}{|c|}{7} & \multicolumn{1}{c|}{7} & \multicolumn{1}{c|}{43} & \multicolumn{1}{c|}{40} & 42.33 & 0.001 & 39.05 & 0.384 & \textbf{39.00} & 0.327 & \multicolumn{1}{c|}{38.95} \\ \cline{1-4} 
\multicolumn{1}{|c|}{7} & \multicolumn{1}{c|}{8} & \multicolumn{1}{c|}{49} & \multicolumn{1}{c|}{40} & 54.03 & 0.001 & 50.20 & 0.912 & \textbf{48.85} & 0.560 & \multicolumn{1}{c|}{$\ast$} \\ \cline{1-4} 
\multicolumn{1}{|c|}{7} & \multicolumn{1}{c|}{9} & \multicolumn{1}{c|}{55} & \multicolumn{1}{c|}{40} & 69.00 & 0.002 & 63.58 & 1.221 & \textbf{61.03} & 0.853 & \multicolumn{1}{c|}{$\ast$} \\ \cline{1-4} 
\multicolumn{1}{|c|}{7} & \multicolumn{1}{c|}{10} & \multicolumn{1}{c|}{61} & \multicolumn{1}{c|}{40} & 81.05 & 0.003 & 75.95 & 1.230 & \textbf{70.15} & 0.607 & \multicolumn{1}{c|}{$\ast$} \\ \cline{1-4} 
\multicolumn{1}{|c|}{7} & \multicolumn{1}{c|}{11} & \multicolumn{1}{c|}{67} & \multicolumn{1}{c|}{40} & 97.33 & 0.000 & 91.98 & 1.259 & \textbf{83.58} & 0.849 & \multicolumn{1}{c|}{$\ast$} \\ \cline{1-4} 
\multicolumn{1}{|c|}{7} & \multicolumn{1}{c|}{12} & \multicolumn{1}{c|}{73} & \multicolumn{1}{c|}{40} & 115.88 & 0.008 & 110.63 & 1.403 & \textbf{98.05} & 0.977 & \multicolumn{1}{c|}{$\ast$} \\ \hline 
\end{tabular}
}
\end{center}
\scriptsize{$^a$ CPU Intel Core i7-3632QM 2.20GHz with 1.0 GB RAM}
\caption{Computational results on the set {\sc group 2} (a)}
\label{tab:g2a}
\end{table}

\begin{table}
\begin{center}
{\scriptsize
\begin{tabular}{cccc|cc|cc|cc|c}
\cline{5-10}
 & & & \multicolumn{1}{c|}{} & \multicolumn{2}{c|}{BL$^a$} & \multicolumn{2}{c|}{BBT$^a$}  & \multicolumn{2}{c|}{BBS$^a$} & \multicolumn{1}{c}{} \\ \hline
\multicolumn{1}{|c|}{$w$} & \multicolumn{1}{c|}{$h$} & \multicolumn{1}{c|}{$n$} & \multicolumn{1}{c|}{{\scriptsize \#}{\sc i}}  & \multicolumn{1}{c|}{Resh} & \multicolumn{1}{c|}{Time} & \multicolumn{1}{c|}{Resh} & \multicolumn{1}{c|}{Time} & \multicolumn{1}{c|}{Resh} & \multicolumn{1}{c|}{Time} & \multicolumn{1}{c|}{Opt}   \\ \hline
\multicolumn{1}{|c|}{8} & \multicolumn{1}{c|}{3} & \multicolumn{1}{c|}{22} & \multicolumn{1}{c|}{40} & \textbf{9.73}$^\blacktriangle$ & 0.000 & \textbf{9.73}$^\blacktriangle$ & 0.002 & \textbf{9.73}$^\blacktriangle$ & 0.004 & \multicolumn{1}{c|}{9.73} \\ \cline{1-4} 
\multicolumn{1}{|c|}{8} & \multicolumn{1}{c|}{4} & \multicolumn{1}{c|}{29} & \multicolumn{1}{c|}{40} & 18.18 & 0.000 & \textbf{17.95}$^\blacktriangle$ & 0.003 & \textbf{17.95}$^\blacktriangle$ & 0.008 & \multicolumn{1}{c|}{17.95} \\ \cline{1-4} 
\multicolumn{1}{|c|}{8} & \multicolumn{1}{c|}{5} & \multicolumn{1}{c|}{36} & \multicolumn{1}{c|}{40} & 26.00 & 0.000 & \textbf{25.40}$^\blacktriangle$ & 0.005 & \textbf{25.40}$^\blacktriangle$ & 0.047 & \multicolumn{1}{c|}{25.40} \\ \cline{1-4} 
\multicolumn{1}{|c|}{8} & \multicolumn{1}{c|}{6} & \multicolumn{1}{c|}{43} & \multicolumn{1}{c|}{40} & 38.20 & 0.000 & 35.75 & 0.173 & \textbf{35.70} & 0.182 & \multicolumn{1}{c|}{35.68} \\ \cline{1-4} 
\multicolumn{1}{|c|}{8} & \multicolumn{1}{c|}{7} & \multicolumn{1}{c|}{50} & \multicolumn{1}{c|}{40} & 48.90 & 0.001 & 44.90 & 0.480 & \textbf{44.63} & 0.451 & \multicolumn{1}{c|}{44.48} \\ \cline{1-4} 
\multicolumn{1}{|c|}{8} & \multicolumn{1}{c|}{8} & \multicolumn{1}{c|}{57} & \multicolumn{1}{c|}{40} & 62.08 & 0.002 & 57.53 & 1.125 & \textbf{55.70} & 0.738 & \multicolumn{1}{c|}{$\ast$} \\ \cline{1-4} 
\multicolumn{1}{|c|}{8} & \multicolumn{1}{c|}{9} & \multicolumn{1}{c|}{64} & \multicolumn{1}{c|}{40} & 77.00 & 0.003 & 72.23 & 1.218 & \textbf{68.45} & 0.633 & \multicolumn{1}{c|}{$\ast$} \\ \cline{1-4} 
\multicolumn{1}{|c|}{8} & \multicolumn{1}{c|}{10} & \multicolumn{1}{c|}{71} & \multicolumn{1}{c|}{40} & 94.60 & 0.000 & 87.83 & 1.244 & \textbf{80.90} & 0.882 & \multicolumn{1}{c|}{$\ast$} \\ \cline{1-4} 
\multicolumn{1}{|c|}{8} & \multicolumn{1}{c|}{11} & \multicolumn{1}{c|}{78} & \multicolumn{1}{c|}{40} & 114.53 & 0.000 & 109.93 & 1.367 & \textbf{99.10} & 0.973 & \multicolumn{1}{c|}{$\ast$} \\ \cline{1-4} 
\multicolumn{1}{|c|}{8} & \multicolumn{1}{c|}{12} & \multicolumn{1}{c|}{85} & \multicolumn{1}{c|}{40} & 136.63 & 0.000 & 131.85 & 1.476 & \textbf{113.50} & 0.984 & \multicolumn{1}{c|}{$\ast$} \\ \cline{1-4} 
\multicolumn{1}{|c|}{9} & \multicolumn{1}{c|}{3} & \multicolumn{1}{c|}{25} & \multicolumn{1}{c|}{40} & \textbf{11.45}$^\blacktriangle$ & 0.000 & \textbf{11.45}$^\blacktriangle$ & 0.002 & \textbf{11.45}$^\blacktriangle$ & 0.004 & \multicolumn{1}{c|}{11.45} \\ \cline{1-4} 
\multicolumn{1}{|c|}{9} & \multicolumn{1}{c|}{4} & \multicolumn{1}{c|}{33} & \multicolumn{1}{c|}{40} & 19.53 & 0.000 & \textbf{19.15}$^\blacktriangle$ & 0.003 & \textbf{19.15}$^\blacktriangle$ & 0.009 & \multicolumn{1}{c|}{19.15} \\ \cline{1-4} 
\multicolumn{1}{|c|}{9} & \multicolumn{1}{c|}{5} & \multicolumn{1}{c|}{41} & \multicolumn{1}{c|}{40} & 29.40 & 0.000 & \textbf{28.65}$^\blacktriangle$ & 0.014 & \textbf{28.65}$^\blacktriangle$ & 0.048 & \multicolumn{1}{c|}{28.65} \\ \cline{1-4} 
\multicolumn{1}{|c|}{9} & \multicolumn{1}{c|}{6} & \multicolumn{1}{c|}{49} & \multicolumn{1}{c|}{40} & 42.43 & 0.001 & 39.80 & 0.292 & \textbf{39.70} & 0.310 & \multicolumn{1}{c|}{39.63} \\ \cline{1-4} 
\multicolumn{1}{|c|}{9} & \multicolumn{1}{c|}{7} & \multicolumn{1}{c|}{57} & \multicolumn{1}{c|}{40} & 55.18 & 0.001 & 51.65 & 0.872 & \textbf{50.73} & 0.645 & \multicolumn{1}{c|}{$\ast$} \\ \cline{1-4} 
\multicolumn{1}{|c|}{9} & \multicolumn{1}{c|}{8} & \multicolumn{1}{c|}{65} & \multicolumn{1}{c|}{40} & 71.33 & 0.003 & 67.70 & 1.185 & \textbf{64.48} & 0.603 & \multicolumn{1}{c|}{$\ast$} \\ \cline{1-4} 
\multicolumn{1}{|c|}{9} & \multicolumn{1}{c|}{9} & \multicolumn{1}{c|}{73} & \multicolumn{1}{c|}{40} & 89.45 & 0.000 & 84.23 & 1.274 & \textbf{78.75} & 0.849 & \multicolumn{1}{c|}{$\ast$} \\ \cline{1-4} 
\multicolumn{1}{|c|}{9} & \multicolumn{1}{c|}{10} & \multicolumn{1}{c|}{81} & \multicolumn{1}{c|}{40} & 108.05 & 0.000 & 103.20 & 1.408 & \textbf{93.80} & 0.811 & \multicolumn{1}{c|}{$\ast$} \\ \cline{1-4} 
\multicolumn{1}{|c|}{9} & \multicolumn{1}{c|}{11} & \multicolumn{1}{c|}{89} & \multicolumn{1}{c|}{40} & 126.45 & 0.000 & 122.63 & 1.392 & \textbf{108.75} & 0.995 & \multicolumn{1}{c|}{$\ast$} \\ \cline{1-4} 
\multicolumn{1}{|c|}{9} & \multicolumn{1}{c|}{12} & \multicolumn{1}{c|}{97} & \multicolumn{1}{c|}{40} & 149.58 & 0.000 & 146.25 & 1.515 & \textbf{128.30} & 1.008 & \multicolumn{1}{c|}{$\ast$} \\ \cline{1-4} 
\multicolumn{1}{|c|}{10} & \multicolumn{1}{c|}{3} & \multicolumn{1}{c|}{28} & \multicolumn{1}{c|}{40} & \textbf{11.88}$^\blacktriangle$ & 0.000 & \textbf{11.88}$^\blacktriangle$ & 0.002 & \textbf{11.88}$^\blacktriangle$ & 0.005 & \multicolumn{1}{c|}{11.88} \\ \cline{1-4} 
\multicolumn{1}{|c|}{10} & \multicolumn{1}{c|}{4} & \multicolumn{1}{c|}{37} & \multicolumn{1}{c|}{40} & 22.75 & 0.000 & \textbf{22.35}$^\blacktriangle$ & 0.003 & \textbf{22.35}$^\blacktriangle$ & 0.009 & \multicolumn{1}{c|}{22.35} \\ \cline{1-4} 
\multicolumn{1}{|c|}{10} & \multicolumn{1}{c|}{5} & \multicolumn{1}{c|}{46} & \multicolumn{1}{c|}{40} & 32.95 & 0.000 & \textbf{31.73} & 0.091 & \textbf{31.73} & 0.115 & \multicolumn{1}{c|}{31.70} \\ \cline{1-4} 
\multicolumn{1}{|c|}{10} & \multicolumn{1}{c|}{6} & \multicolumn{1}{c|}{55} & \multicolumn{1}{c|}{40} & 46.03 & 0.001 & 43.93 & 0.448 & \textbf{43.58} & 0.344 & \multicolumn{1}{c|}{43.55} \\ \cline{1-4} 
\multicolumn{1}{|c|}{10} & \multicolumn{1}{c|}{7} & \multicolumn{1}{c|}{64} & \multicolumn{1}{c|}{40} & 62.28 & 0.002 & 58.08 & 1.053 & \textbf{56.38} & 0.482 & \multicolumn{1}{c|}{$\ast$} \\ \cline{1-4} 
\multicolumn{1}{|c|}{10} & \multicolumn{1}{c|}{8} & \multicolumn{1}{c|}{73} & \multicolumn{1}{c|}{40} & 77.25 & 0.004 & 72.75 & 1.375 & \textbf{69.28} & 0.730 & \multicolumn{1}{c|}{$\ast$} \\ \cline{1-4} 
\multicolumn{1}{|c|}{10} & \multicolumn{1}{c|}{9} & \multicolumn{1}{c|}{82} & \multicolumn{1}{c|}{40} & 97.18 & 0.007 & 93.00 & 1.493 & \textbf{86.43} & 0.782 & \multicolumn{1}{c|}{$\ast$} \\ \cline{1-4} 
\multicolumn{1}{|c|}{10} & \multicolumn{1}{c|}{10} & \multicolumn{1}{c|}{91} & \multicolumn{1}{c|}{40} & 117.65 & 0.012 & 114.10 & 1.501 & \textbf{102.40} & 0.978 & \multicolumn{1}{c|}{$\ast$} \\ \cline{1-4} 
\multicolumn{1}{|c|}{10} & \multicolumn{1}{c|}{11} & \multicolumn{1}{c|}{100} & \multicolumn{1}{c|}{40} & 140.25 & 0.000 & 138.63 & 1.502 & \textbf{121.55} & 0.727 & \multicolumn{1}{c|}{$\ast$} \\ \cline{1-4} 
\multicolumn{1}{|c|}{10} & \multicolumn{1}{c|}{12} & \multicolumn{1}{c|}{109} & \multicolumn{1}{c|}{40} & 160.08 & 0.000 & 158.08 & 1.470 & \textbf{137.68} & 0.896 & \multicolumn{1}{c|}{$\ast$} \\ \cline{1-4} 
\multicolumn{1}{|c|}{11} & \multicolumn{1}{c|}{3} & \multicolumn{1}{c|}{31} & \multicolumn{1}{c|}{40} & 14.23 & 0.000 & \textbf{14.13}$^\blacktriangle$ & 0.002 & \textbf{14.13}$^\blacktriangle$ & 0.006 & \multicolumn{1}{c|}{14.13} \\ \cline{1-4} 
\multicolumn{1}{|c|}{11} & \multicolumn{1}{c|}{4} & \multicolumn{1}{c|}{41} & \multicolumn{1}{c|}{40} & 23.90 & 0.000 & \textbf{23.43}$^\blacktriangle$ & 0.004 & \textbf{23.43}$^\blacktriangle$ & 0.012 & \multicolumn{1}{c|}{23.43} \\ \cline{1-4} 
\multicolumn{1}{|c|}{11} & \multicolumn{1}{c|}{5} & \multicolumn{1}{c|}{51} & \multicolumn{1}{c|}{40} & 36.13 & 0.000 & \textbf{35.05}$^\blacktriangle$ & 0.044 & \textbf{35.05}$^\blacktriangle$ & 0.135 & \multicolumn{1}{c|}{35.05} \\ \cline{1-4} 
\multicolumn{1}{|c|}{11} & \multicolumn{1}{c|}{6} & \multicolumn{1}{c|}{61} & \multicolumn{1}{c|}{40} & 50.30 & 0.001 & 48.00 & 0.783 & \textbf{47.13} & 0.329 & \multicolumn{1}{c|}{$\ast$} \\ \cline{1-4} 
\multicolumn{1}{|c|}{11} & \multicolumn{1}{c|}{7} & \multicolumn{1}{c|}{71} & \multicolumn{1}{c|}{40} & 68.23 & 0.003 & 64.75 & 1.260 & \textbf{62.33} & 0.584 & \multicolumn{1}{c|}{$\ast$} \\ \cline{1-4} 
\multicolumn{1}{|c|}{11} & \multicolumn{1}{c|}{8} & \multicolumn{1}{c|}{81} & \multicolumn{1}{c|}{40} & 87.45 & 0.005 & 83.95 & 1.438 & \textbf{79.13} & 0.601 & \multicolumn{1}{c|}{$\ast$} \\ \cline{1-4} 
\multicolumn{1}{|c|}{11} & \multicolumn{1}{c|}{9} & \multicolumn{1}{c|}{91} & \multicolumn{1}{c|}{40} & 110.15 & 0.010 & 105.75 & 1.475 & \textbf{96.75} & 0.918 & \multicolumn{1}{c|}{$\ast$} \\ \cline{1-4} 
\multicolumn{1}{|c|}{11} & \multicolumn{1}{c|}{10} & \multicolumn{1}{c|}{101} & \multicolumn{1}{c|}{40} & 130.95 & 0.016 & 127.78 & 1.506 & \textbf{115.15} & 0.639 & \multicolumn{1}{c|}{$\ast$} \\ \cline{1-4} 
\multicolumn{1}{|c|}{11} & \multicolumn{1}{c|}{11} & \multicolumn{1}{c|}{111} & \multicolumn{1}{c|}{40} & 153.33 & 0.027 & 150.85 & 1.367 & \textbf{133.03} & 0.891 & \multicolumn{1}{c|}{$\ast$} \\ \cline{1-4} 
\multicolumn{1}{|c|}{11} & \multicolumn{1}{c|}{12} & \multicolumn{1}{c|}{121} & \multicolumn{1}{c|}{40} & 180.15 & 0.000 & 179.15 & 1.516 & \textbf{156.05} & 0.621 & \multicolumn{1}{c|}{$\ast$} \\ \cline{1-4} 
\multicolumn{1}{|c|}{12} & \multicolumn{1}{c|}{3} & \multicolumn{1}{c|}{34} & \multicolumn{1}{c|}{40} & 14.95 & 0.000 & \textbf{14.90}$^\blacktriangle$ & 0.002 & \textbf{14.90}$^\blacktriangle$ & 0.006 & \multicolumn{1}{c|}{14.90} \\ \cline{1-4} 
\multicolumn{1}{|c|}{12} & \multicolumn{1}{c|}{4} & \multicolumn{1}{c|}{45} & \multicolumn{1}{c|}{40} & 27.15 & 0.000 & \textbf{26.80}$^\blacktriangle$ & 0.004 & \textbf{26.80}$^\blacktriangle$ & 0.023 & \multicolumn{1}{c|}{26.80} \\ \cline{1-4} 
\multicolumn{1}{|c|}{12} & \multicolumn{1}{c|}{5} & \multicolumn{1}{c|}{56} & \multicolumn{1}{c|}{40} & 40.48 & 0.001 & 39.10 & 0.309 & \textbf{39.05} & 0.282 & \multicolumn{1}{c|}{39.00} \\ \cline{1-4} 
\multicolumn{1}{|c|}{12} & \multicolumn{1}{c|}{6} & \multicolumn{1}{c|}{67} & \multicolumn{1}{c|}{40} & 55.40 & 0.001 & 52.43 & 0.748 & \textbf{51.75} & 0.360 & \multicolumn{1}{c|}{$\ast$} \\ \cline{1-4} 
\multicolumn{1}{|c|}{12} & \multicolumn{1}{c|}{7} & \multicolumn{1}{c|}{78} & \multicolumn{1}{c|}{40} & 74.75 & 0.003 & 71.60 & 1.446 & \textbf{69.15} & 0.809 & \multicolumn{1}{c|}{$\ast$} \\ \cline{1-4} 
\multicolumn{1}{|c|}{12} & \multicolumn{1}{c|}{8} & \multicolumn{1}{c|}{89} & \multicolumn{1}{c|}{40} & 95.65 & 0.007 & 92.18 & 1.500 & \textbf{86.45} & 0.753 & \multicolumn{1}{c|}{$\ast$} \\ \cline{1-4} 
\multicolumn{1}{|c|}{12} & \multicolumn{1}{c|}{9} & \multicolumn{1}{c|}{100} & \multicolumn{1}{c|}{40} & 117.63 & 0.013 & 112.85 & 1.548 & \textbf{104.25} & 0.591 & \multicolumn{1}{c|}{$\ast$} \\ \cline{1-4} 
\multicolumn{1}{|c|}{12} & \multicolumn{1}{c|}{10} & \multicolumn{1}{c|}{111} & \multicolumn{1}{c|}{40} & 137.40 & 0.000 & 135.63 & 1.524 & \textbf{121.25} & 0.812 & \multicolumn{1}{c|}{$\ast$} \\ \cline{1-4} 
\multicolumn{1}{|c|}{12} & \multicolumn{1}{c|}{11} & \multicolumn{1}{c|}{122} & \multicolumn{1}{c|}{40} & 170.60 & 0.000 & 169.35 & 1.506 & \textbf{147.55} & 0.580 & \multicolumn{1}{c|}{$\ast$} \\ \cline{1-4} 
\multicolumn{1}{|c|}{12} & \multicolumn{1}{c|}{12} & \multicolumn{1}{c|}{133} & \multicolumn{1}{c|}{40} & 198.80 & 0.000 & 198.10 & 1.685 & \textbf{171.68} & 0.849 & \multicolumn{1}{c|}{$\ast$} \\ \hline 
\end{tabular}
}
\end{center}
\scriptsize{$^a$ CPU Intel Core i7-3632QM 2.20GHz with 1.0 GB RAM}
\caption{Computational results on the set {\sc group 2} (b)}
\label{tab:g2b}
\end{table}

\begin{figure}[htbp] 
  \begin{center}
\includegraphics[width=7cm,height=5cm]{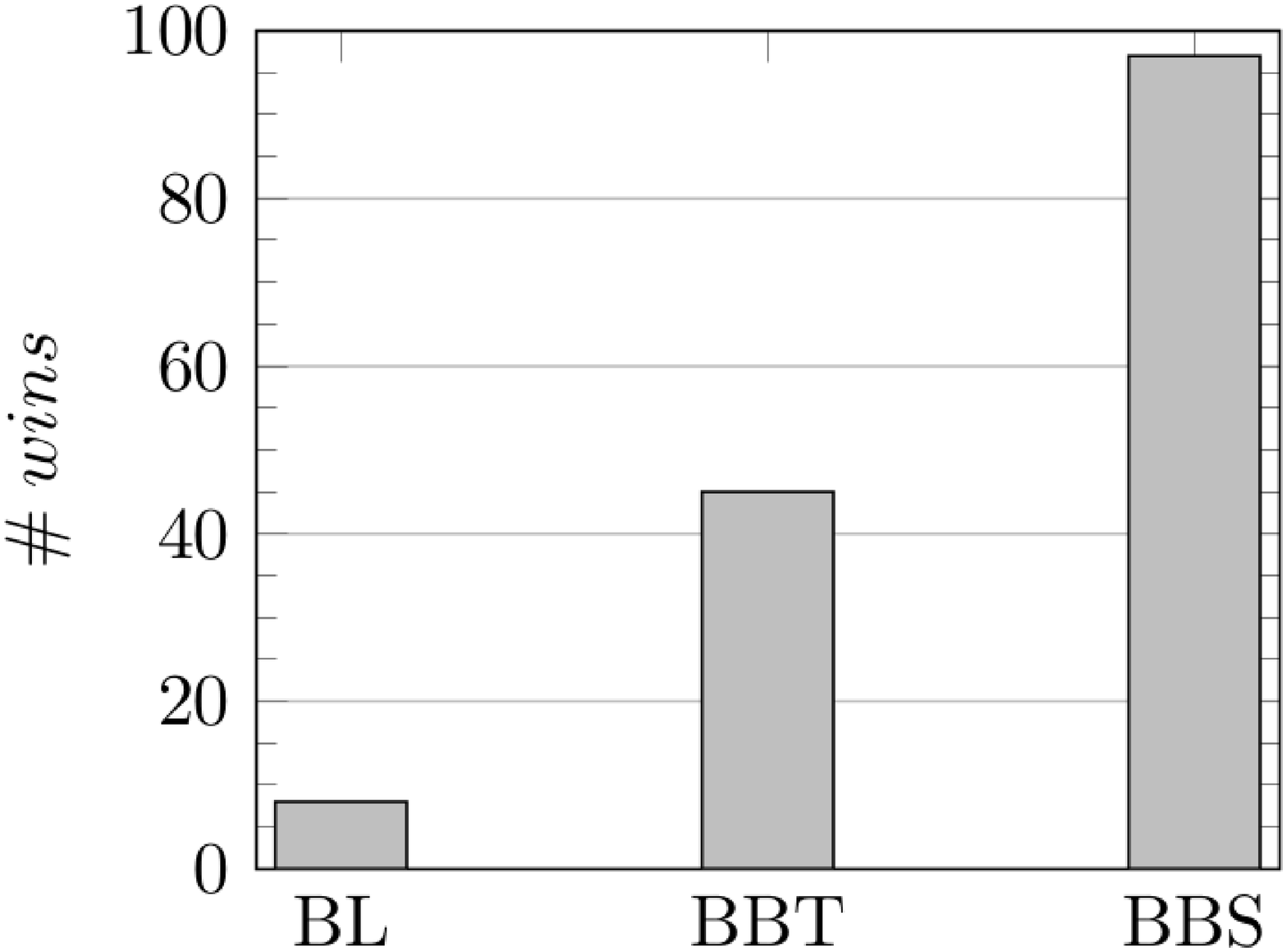}
    \end{center}
\caption{Aggregated results on the set {\sc group 2}}
\label{fig:g2}
\end{figure}

In Table \ref{tab:g3} and Figure \ref{fig:g3} the computational results obtained on the instances
of {\sc group 3} are shown.
Column {\sc id} indicates if the instances are Random (R) or Upside down (U), while column b reports the number of bays.
Symbol ``-'' in Columns Resh and Time for the iterative deepening A* restricted by \citet{ZQLZ2012} (IDA$^*$-R) means that the value of the solution and the computing time are not available in the literature.
For most of the instances, the Bounded Beam Search approach produces the best
results. Observe that the R instances turn out to be
harder to solve than the U ones. 
 
\begin{table}
\begin{center}
{\fontsize{6}{8}\selectfont
\begin{tabular}{cccccc|cc|cc|cc|cc|cc|c|}
\cline{7-16}
 & & & & & \multicolumn{1}{c|}{} & \multicolumn{2}{c|}{3PH$^a$} & \multicolumn{2}{c|}{IDA$^*$-R$^b$} & \multicolumn{2}{c|}{BL$^c$} & \multicolumn{2}{c|}{BBT$^c$} & \multicolumn{2}{c|}{BBS$^c$} & \multicolumn{1}{c}{} \\ \hline
\multicolumn{1}{|c|}{{\sc id}}          & \multicolumn{1}{c|}{$b$} & \multicolumn{1}{c|}{$w$} & \multicolumn{1}{c|}{$h$} & \multicolumn{1}{c|}{$n$} & {\scriptsize \#}{\sc i} & \multicolumn{1}{c|}{Resh} & \multicolumn{1}{c|}{Time} & \multicolumn{1}{c|}{Resh} & \multicolumn{1}{c|}{Time}  & \multicolumn{1}{c|}{Resh} & \multicolumn{1}{c|}{Time} & \multicolumn{1}{c|}{Resh} & \multicolumn{1}{c|}{Time} & \multicolumn{1}{c|}{Resh} & \multicolumn{1}{c|}{Time} & \multicolumn{1}{c|}{Opt}   \\ \hline
\multicolumn{1}{|c|}{R} & \multicolumn{1}{c|}{1} & \multicolumn{1}{c|}{16} & \multicolumn{1}{c|}{6} & \multicolumn{1}{c|}{70} & \multicolumn{1}{c|}{5} & 55.40 & 8204.32 & 40.00 & - & 39.80 & 0.00 & 40.00 & 0.59 & \textbf{39.60} & 0.60 & \multicolumn{1}{c|}{$\ast$} \\ \cline{1-6} 
\multicolumn{1}{|c|}{R} & \multicolumn{1}{c|}{1} & \multicolumn{1}{c|}{16} & \multicolumn{1}{c|}{8} & \multicolumn{1}{c|}{90} & \multicolumn{1}{c|}{5} & 101.40 & 13353.36  & 63.00 & - & 63.00 & 0.01 & 63.40 & 2.63 & \textbf{61.40} & 0.75 & \multicolumn{1}{c|}{$\ast$} \\ \cline{1-6} 
\multicolumn{1}{|c|}{R} & \multicolumn{1}{c|}{2} & \multicolumn{1}{c|}{32} & \multicolumn{1}{c|}{6} & \multicolumn{1}{c|}{140} & \multicolumn{1}{c|}{5} & 90.20 & 21579.55  & - & - & 71.80 & 0.01 & 71.80 & 0.45 & \textbf{71.60} & 0.11 & \multicolumn{1}{c|}{$\ast$} \\ \cline{1-6} 
\multicolumn{1}{|c|}{R} & \multicolumn{1}{c|}{2} & \multicolumn{1}{c|}{32} & \multicolumn{1}{c|}{8} & \multicolumn{1}{c|}{190} & \multicolumn{1}{c|}{5} & 177.80 & 21527.71  & - & - & 125.40 & 0.00 & 125.40 & 9.16 & \textbf{124.00} & 0.80 & \multicolumn{1}{c|}{$\ast$} \\ \cline{1-6} 
\multicolumn{1}{|c|}{R} & \multicolumn{1}{c|}{4} & \multicolumn{1}{c|}{64} & \multicolumn{1}{c|}{6} & \multicolumn{1}{c|}{280} & \multicolumn{1}{c|}{5} & 174.20 & 21522.11  & - & - & \textbf{151.60} & 0.09 & 152.60 & 71.65 & 152.40 & 0.60 & \multicolumn{1}{c|}{$\ast$} \\ \cline{1-6} 
\multicolumn{1}{|c|}{R} & \multicolumn{1}{c|}{4} & \multicolumn{1}{c|}{64} & \multicolumn{1}{c|}{8} & \multicolumn{1}{c|}{380} & \multicolumn{1}{c|}{5} & 389.20 & 21231.08  & - & - & \textbf{242.80} & 0.00 & 245.40 & 9.39 & 243.60 & 1.00 & \multicolumn{1}{c|}{$\ast$} \\ \cline{1-6} 
\multicolumn{1}{|c|}{R} & \multicolumn{1}{c|}{6} & \multicolumn{1}{c|}{96} & \multicolumn{1}{c|}{6} & \multicolumn{1}{c|}{430} & \multicolumn{1}{c|}{5} & 279.8 & 21358.80  & - & - & \textbf{226.60} & 0.00 & \textbf{226.60} & 6.45 & \textbf{226.60} & 0.40 & \multicolumn{1}{c|}{$\ast$} \\ \cline{1-6} 
\multicolumn{1}{|c|}{R} & \multicolumn{1}{c|}{6} & \multicolumn{1}{c|}{96} & \multicolumn{1}{c|}{8} & \multicolumn{1}{c|}{570} & \multicolumn{1}{c|}{5} & 672.00 & 20874.03  & - & - & 352.40 & 0.00 & 353.60 & 116.88 & \textbf{352.00} & 0.60 & \multicolumn{1}{c|}{$\ast$} \\ \cline{1-6} 
\multicolumn{1}{|c|}{R} & \multicolumn{1}{c|}{8} & \multicolumn{1}{c|}{128} & \multicolumn{1}{c|}{6} & \multicolumn{1}{c|}{570} & \multicolumn{1}{c|}{5} & 375.40 & 21200.97  & - & - & \textbf{305.60}$^\blacktriangle$ & 0.00 & \textbf{305.60}$^\blacktriangle$ & 0.16 & \textbf{305.60}$^\blacktriangle$ & 0.00 & \multicolumn{1}{c|}{305.60} \\ \cline{1-6} 
\multicolumn{1}{|c|}{R} & \multicolumn{1}{c|}{10} & \multicolumn{1}{c|}{160} & \multicolumn{1}{c|}{6} & \multicolumn{1}{c|}{720} & \multicolumn{1}{c|}{5} & 449.20 & 20844.76  & - & - & \textbf{372.60} & 0.22 & 372.80 & 7.20 & 372.80 & 0.21 & \multicolumn{1}{c|}{$\ast$} \\ \cline{1-6} 
\multicolumn{1}{|c|}{U} & \multicolumn{1}{c|}{1} & \multicolumn{1}{c|}{16} & \multicolumn{1}{c|}{6} & \multicolumn{1}{c|}{70} & \multicolumn{1}{c|}{2} & 57.50 & 14284.86  & \textbf{56.50}$^\blacktriangle$ & -  & \textbf{56.50}$^\blacktriangle$ & 0.00 & \textbf{56.50}$^\blacktriangle$ & 0.00 & \textbf{56.50}$^\blacktriangle$ & 0.00 & \multicolumn{1}{c|}{56.50} \\ \cline{1-6} 
\multicolumn{1}{|c|}{U} & \multicolumn{1}{c|}{1} & \multicolumn{1}{c|}{16} & \multicolumn{1}{c|}{8} & \multicolumn{1}{c|}{90} & \multicolumn{1}{c|}{2} & 87.50 & 14804.06  & \textbf{77.00}$^\blacktriangle$ & - & \textbf{77.00}$^\blacktriangle$ & 0.00 & \textbf{77.00}$^\blacktriangle$ & 0.00 & \textbf{77.00}$^\blacktriangle$ & 0.00 & \multicolumn{1}{c|}{77.00} \\ \cline{1-6} 
\multicolumn{1}{|c|}{U} & \multicolumn{1}{c|}{2} & \multicolumn{1}{c|}{32} & \multicolumn{1}{c|}{6} & \multicolumn{1}{c|}{140} & \multicolumn{1}{c|}{2} & 117.00 & 21543.79  & - & - & \textbf{109.00}$^\blacktriangle$ & 0.00 & \textbf{109.00}$^\blacktriangle$ & 0.04 & \textbf{109.00}$^\blacktriangle$ & 0.00 & \multicolumn{1}{c|}{109.00} \\ \cline{1-6} 
\multicolumn{1}{|c|}{U} & \multicolumn{1}{c|}{2} & \multicolumn{1}{c|}{32} & \multicolumn{1}{c|}{8} & \multicolumn{1}{c|}{190} & \multicolumn{1}{c|}{2} & 210.50 & 21371.03  & - & - & \textbf{160.00}$^\blacktriangle$ & 0.01 & \textbf{160.00}$^\blacktriangle$ & 0.06 & \textbf{160.00}$^\blacktriangle$ & 0.00 & \multicolumn{1}{c|}{160.00} \\ \cline{1-6} 
\multicolumn{1}{|c|}{U} & \multicolumn{1}{c|}{4} & \multicolumn{1}{c|}{64} & \multicolumn{1}{c|}{6} & \multicolumn{1}{c|}{280} & \multicolumn{1}{c|}{2} & 227.50 & 21321.99  & - & -  & \textbf{216.50}$^\blacktriangle$ & 0.00 & \textbf{216.50}$^\blacktriangle$ & 0.07 & \textbf{216.50}$^\blacktriangle$ & 0.00 & \multicolumn{1}{c|}{216.50} \\ \cline{1-6} 
\multicolumn{1}{|c|}{U} & \multicolumn{1}{c|}{4} & \multicolumn{1}{c|}{64} & \multicolumn{1}{c|}{8} & \multicolumn{1}{c|}{380} & \multicolumn{1}{c|}{2} & 420.00 & 20890.37  & - & - & \textbf{318.50}$^\blacktriangle$ & 0.00 & \textbf{318.50}$^\blacktriangle$ & 0.07 & \textbf{318.50}$^\blacktriangle$ & 0.00 & \multicolumn{1}{c|}{318.50} \\ \cline{1-6} 
\multicolumn{1}{|c|}{U} & \multicolumn{1}{c|}{6} & \multicolumn{1}{c|}{96} & \multicolumn{1}{c|}{6} & \multicolumn{1}{c|}{430} & \multicolumn{1}{c|}{2} & 363.00 & 20929.72  & - & - & \textbf{334.00}$^\blacktriangle$ & 0.00 & \textbf{334.00}$^\blacktriangle$ & 0.17 & \textbf{334.00}$^\blacktriangle$ & 0.00 & \multicolumn{1}{c|}{334.00} \\ \cline{1-6} 
\multicolumn{1}{|c|}{U} & \multicolumn{1}{c|}{6} & \multicolumn{1}{c|}{96} & \multicolumn{1}{c|}{8} & \multicolumn{1}{c|}{570} & \multicolumn{1}{c|}{2} & 664.00 & 20066.65  & - & - & \textbf{476.00}$^\blacktriangle$ & 0.00 & \textbf{476.00}$^\blacktriangle$ & 0.21 & \textbf{476.00}$^\blacktriangle$ & 0.00 & \multicolumn{1}{c|}{476.00} \\ \cline{1-6} 
\multicolumn{1}{|c|}{U} & \multicolumn{1}{c|}{8} & \multicolumn{1}{c|}{128} & \multicolumn{1}{c|}{6} & \multicolumn{1}{c|}{570} & \multicolumn{1}{c|}{2} & 474.00 & 20504.17  & - & - & \textbf{443.50}$^\blacktriangle$ & 0.00 & \textbf{443.50}$^\blacktriangle$ & 0.26 & \textbf{443.50}$^\blacktriangle$ & 0.01 & \multicolumn{1}{c|}{443.50} \\ \cline{1-6} 
\multicolumn{1}{|c|}{U} & \multicolumn{1}{c|}{10} & \multicolumn{1}{c|}{160} & \multicolumn{1}{c|}{6} & \multicolumn{1}{c|}{720} & \multicolumn{1}{c|}{2} & 625.00 & 19722.45  & - & - & \textbf{561.00}$^\blacktriangle$ & 0.00 & \textbf{561.00}$^\blacktriangle$ & 0.32 & \textbf{561.00}$^\blacktriangle$ & 0.01 & \multicolumn{1}{c|}{561.00} \\ \hline 
\end{tabular}
}
\end{center}
{\scriptsize $^a$Core 2 Duo E8500 3.42GHz and 3.46 GB RAM}\\
{\scriptsize $^b$Intel Core i7-920 2.66GHz and 12 GB RAM}\\
{\scriptsize $^c$Intel Core i7-3632QM 2.2GHz and 1 GB RAM}
\caption{Computational results on the set {\sc group 3}}
\label{tab:g3}
\end{table}

\begin{figure}[htbp] 
  \begin{center}
\includegraphics[width=7cm,height=5cm]{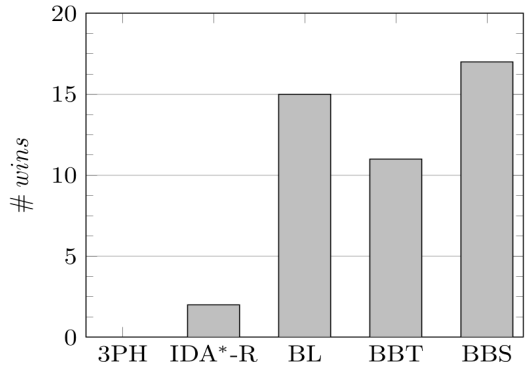}

    \end{center}
\caption{Aggregated results on the set {\sc group 3}}
\label{fig:g3}
\end{figure}

Table \ref{tab:g4} and Figure \ref{fig:g4} report the computational results on the instances of {\sc group 4}.
Here, we did not consider the results in \citet{CVS2011} since, as reported in \citet{ZQLZ2012}, they contain some errors.
Symbol ``-'' in Columns Resh and Time for the Matrix-Algorithm by \citet{CSV2009} (MA) and the iterative deepening A* restricted by \citet{ZQLZ2012} (IDA$^*$-R) means that the value of the solution and the computing time are not available in the literature.
Here, the Bounded Beam Search method is always the best approach and in most of the cases it finds the optimal solution.

\begin{table}
\begin{center}
{\fontsize{7}{10}\selectfont
\begin{tabular}{cccc|cc|cc|cc|cc|cc|cc|cc|c}
\cline{5-14}
 & & & \multicolumn{1}{c|}{} & \multicolumn{2}{c|}{MA$^a$} & \multicolumn{2}{c|}{IDA$^*$-R$^b$}   & \multicolumn{2}{c|}{BL$^c$} & \multicolumn{2}{c|}{BBT$^c$} & \multicolumn{2}{c|}{BBS$^c$} & \multicolumn{1}{c}{} \\ \hline
\multicolumn{1}{|c|}{$w$} & \multicolumn{1}{c|}{$h$} & \multicolumn{1}{c|}{$n$} & \multicolumn{1}{c|}{{\scriptsize \#}{\sc i}}  & \multicolumn{1}{c|}{Resh} & \multicolumn{1}{c|}{Time} & \multicolumn{1}{c|}{Resh} & \multicolumn{1}{c|}{Time} & \multicolumn{1}{c|}{Resh} & \multicolumn{1}{c|}{Time} & \multicolumn{1}{c|}{Resh} & \multicolumn{1}{c|}{Time} & \multicolumn{1}{c|}{Resh} & \multicolumn{1}{c|}{Time} & \multicolumn{1}{c|}{Opt}   \\ \hline
\multicolumn{1}{|c|}{3} & \multicolumn{1}{c|}{5} & \multicolumn{1}{c|}{9} & \multicolumn{1}{c|}{40} & - & - & \textbf{5.000}$^\blacktriangle$ & - & 5.025 & 0.000 & \textbf{5.000}$^\blacktriangle$ & 0.003 & \textbf{5.000}$^\blacktriangle$ & 0.002 & \multicolumn{1}{c|}{5.000} \\ \cline{1-4} 
\multicolumn{1}{|c|}{3} & \multicolumn{1}{c|}{6} & \multicolumn{1}{c|}{12} & \multicolumn{1}{c|}{40} & - & -  & \textbf{6.175}$^\blacktriangle$ & - & 6.225 & 0.000 & \textbf{6.175}$^\blacktriangle$ & 0.002 & \textbf{6.175}$^\blacktriangle$ & 0.002 & \multicolumn{1}{c|}{6.175} \\ \cline{1-4} 
\multicolumn{1}{|c|}{3} & \multicolumn{1}{c|}{7} & \multicolumn{1}{c|}{15} & \multicolumn{1}{c|}{40} & - & -  & \textbf{7.025}$^\blacktriangle$ & - & \textbf{7.025}$^\blacktriangle$ & 0.000 & \textbf{7.025}$^\blacktriangle$ & 0.002 & \textbf{7.025}$^\blacktriangle$ & 0.003 & \multicolumn{1}{c|}{7.025} \\ \cline{1-4} 
\multicolumn{1}{|c|}{3} & \multicolumn{1}{c|}{8} & \multicolumn{1}{c|}{18} & \multicolumn{1}{c|}{40} & - & -  & \textbf{8.400}$^\blacktriangle$ & - & 8.425 & 0.000 & \textbf{8.400}$^\blacktriangle$ & 0.002 & \textbf{8.400}$^\blacktriangle$ & 0.004 & \multicolumn{1}{c|}{8.400} \\ \cline{1-4} 
\multicolumn{1}{|c|}{3} & \multicolumn{1}{c|}{9} & \multicolumn{1}{c|}{21} & \multicolumn{1}{c|}{40} & - & -  & \textbf{9.275}$^\blacktriangle$ & - & 9.325 & 0.000 & \textbf{9.275}$^\blacktriangle$ & 0.002 & \textbf{9.275}$^\blacktriangle$ & 0.004 & \multicolumn{1}{c|}{9.275} \\ \cline{1-4} 
\multicolumn{1}{|c|}{3} & \multicolumn{1}{c|}{10} & \multicolumn{1}{c|}{24} & \multicolumn{1}{c|}{40} & - & -  & \textbf{10.650}$^\blacktriangle$ & - & 10.700 & 0.000 & \textbf{10.650}$^\blacktriangle$ & 0.002 & \textbf{10.650}$^\blacktriangle$ & 0.005 & \multicolumn{1}{c|}{10.650} \\ \cline{1-4} 
\multicolumn{1}{|c|}{4} & \multicolumn{1}{c|}{6} & \multicolumn{1}{c|}{16} & \multicolumn{1}{c|}{40} & - & -  & \textbf{10.200}$^\blacktriangle$ & - & 10.425 & 0.000 & \textbf{10.200}$^\blacktriangle$ & 0.002 & \textbf{10.200}$^\blacktriangle$ & 0.003 & \multicolumn{1}{c|}{10.200} \\ \cline{1-4} 
\multicolumn{1}{|c|}{4} & \multicolumn{1}{c|}{7} & \multicolumn{1}{c|}{20} & \multicolumn{1}{c|}{40} & - & -  & \textbf{12.950}$^\blacktriangle$ & - & 13.175 & 0.000 & \textbf{12.950}$^\blacktriangle$ & 0.002 & \textbf{12.950}$^\blacktriangle$ & 0.010 & \multicolumn{1}{c|}{12.950} \\ \cline{1-4} 
\multicolumn{1}{|c|}{4} & \multicolumn{1}{c|}{8} & \multicolumn{1}{c|}{24} & \multicolumn{1}{c|}{40} & - & -  & \textbf{14.025}$^\blacktriangle$ & - & 14.300 & 0.000 & \textbf{14.025}$^\blacktriangle$ & 0.003 & \textbf{14.025}$^\blacktriangle$ & 0.033 & \multicolumn{1}{c|}{14.025} \\ \cline{1-4} 
\multicolumn{1}{|c|}{4} & \multicolumn{1}{c|}{9} & \multicolumn{1}{c|}{28} & \multicolumn{1}{c|}{40} & - & -  & \textbf{16.125}$^\blacktriangle$ & - & 16.525 & 0.000 & \textbf{16.125}$^\blacktriangle$ & 0.003 & \textbf{16.125}$^\blacktriangle$ & 0.046 & \multicolumn{1}{c|}{16.125} \\ \cline{1-4} 
\multicolumn{1}{|c|}{5} & \multicolumn{1}{c|}{6} & \multicolumn{1}{c|}{20} & \multicolumn{1}{c|}{40} & - & -  & \textbf{15.425}$^\blacktriangle$ & - &16.275 & 0.000 & \textbf{15.425}$^\blacktriangle$ & 0.003 & \textbf{15.425}$^\blacktriangle$ & 0.010 & \multicolumn{1}{c|}{15.425} \\ \cline{1-4} 
\multicolumn{1}{|c|}{5} & \multicolumn{1}{c|}{7} & \multicolumn{1}{c|}{25} & \multicolumn{1}{c|}{40} & - & -  & \textbf{18.850}$^\blacktriangle$ & - &19.850 & 0.000 & \textbf{18.850}$^\blacktriangle$ & 0.005 & \textbf{18.850}$^\blacktriangle$ & 0.065 & \multicolumn{1}{c|}{18.850} \\ \cline{1-4} 
\multicolumn{1}{|c|}{5} & \multicolumn{1}{c|}{8} & \multicolumn{1}{c|}{30} & \multicolumn{1}{c|}{40} & - & -  & \textbf{22.075}$^\blacktriangle$ & - &23.525 & 0.000 & \textbf{22.075}$^\blacktriangle$ & 0.041 & \textbf{22.075}$^\blacktriangle$ & 0.172 & \multicolumn{1}{c|}{22.075} \\ \cline{1-4} 
\multicolumn{1}{|c|}{5} & \multicolumn{1}{c|}{9} & \multicolumn{1}{c|}{35} & \multicolumn{1}{c|}{40} & - & -  & 24.300 & - & 25.375 & 0.000 & 24.400 & 0.187 & \textbf{24.250}$^\blacktriangle$ & 0.318 & \multicolumn{1}{c|}{24.250} \\ \cline{1-4} 
\multicolumn{1}{|c|}{5} & \multicolumn{1}{c|}{10} & \multicolumn{1}{c|}{40} & \multicolumn{1}{c|}{40} & - & - & 27.850 & - & 28.650 & 0.000 & 28.000 & 0.316 & \textbf{27.750} & 0.371 & \multicolumn{1}{c|}{27.700} \\ \cline{1-4} 
\multicolumn{1}{|c|}{5} & \multicolumn{1}{c|}{11} & \multicolumn{1}{c|}{45} & \multicolumn{1}{c|}{40} & - & - & 30.675 & - & 31.350 & 0.001 & 30.900 & 0.513 & \textbf{30.525} & 0.599 & \multicolumn{1}{c|}{30.450} \\ \cline{1-4} 
\multicolumn{1}{|c|}{5} & \multicolumn{1}{c|}{12} & \multicolumn{1}{c|}{50} & \multicolumn{1}{c|}{40} & - & - & 33.625 & - & 34.300 & 0.001 & 34.350 & 0.642 & \textbf{33.325} & 0.667 & \multicolumn{1}{c|}{$\ast$} \\ \cline{1-4} 
\multicolumn{1}{|c|}{6} & \multicolumn{1}{c|}{8} & \multicolumn{1}{c|}{36} & \multicolumn{1}{c|}{40} & 31.8 & 0.38 & 31.075 & - & 34.050 & 0.001 & 31.050 & 0.447 & \textbf{30.925} & 0.566 & \multicolumn{1}{c|}{$\ast$} \\ \cline{1-4} 
\multicolumn{1}{|c|}{6} & \multicolumn{1}{c|}{12} & \multicolumn{1}{c|}{60} & \multicolumn{1}{c|}{40} & 47.6 & 0.65 & 47.200 & - & 48.650 & 0.003 & 48.100 & 1.211 & \textbf{46.000} & 0.836 & \multicolumn{1}{c|}{$\ast$} \\ \cline{1-4} 
\multicolumn{1}{|c|}{10} & \multicolumn{1}{c|}{8} & \multicolumn{1}{c|}{60} & \multicolumn{1}{c|}{40} & 82.9 & 0.93 & 84.975 & - & 91.725 & 0.000 & 85.700 & 1.314 & \textbf{77.300} & 0.788 & \multicolumn{1}{c|}{$\ast$} \\ \cline{1-4} 
\multicolumn{1}{|c|}{10} & \multicolumn{1}{c|}{12} & \multicolumn{1}{c|}{100} & \multicolumn{1}{c|}{40} & 121.3 & 1.57 & 126.325 & - & 128.350 & 0.000 & 132.025 & 1.927 & \textbf{112.475} & 0.917 & \multicolumn{1}{c|}{$\ast$} \\ \cline{1-4} 
\multicolumn{1}{|c|}{100} & \multicolumn{1}{c|}{102} & \multicolumn{1}{c|}{10000} & \multicolumn{1}{c|}{40} & - & - & - & - & 63358.246 & 0.190 & $\ast$ & $\ast$ & \textbf{62413.277} & 1.111 & \multicolumn{1}{c|}{$\ast$} \\ \hline 
\end{tabular}
}
\end{center}
{\scriptsize $^a$Pentium IV and 512 MB RAM}\\
{\scriptsize $^b$Intel Core i7-920 2.66GHz and 12 GB RAM}\\
{\scriptsize $^c$Intel Core i7-3632QM 2.2GHz and 1 GB RAM}
\caption{Computational results on the set {\sc group 4}}
\label{tab:g4}
\end{table}

\begin{figure}[htbp] 
  \begin{center}
\includegraphics[width=7cm,height=6cm]{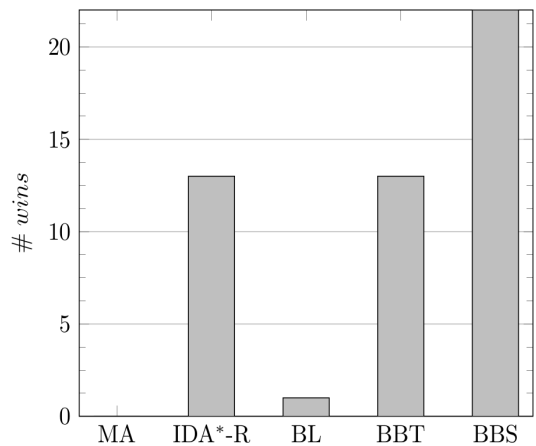}
    \end{center}
\caption{Aggregate results on the set {\sc group 4}}
\label{fig:g4}
\end{figure}

In Table \ref{tab:g5} and Figure \ref{fig:g5}, the computational results obtained on the instances of {\sc group 5} are shown.
Column {\sc type} denotes if the instances are balanced (b) or unbalanced (u) and column d reports the storage density.
For the sake of conciseness, in the table we report only the rows where at least
one between BBT and the Bounded Beam Search methods does not find the optimal
solution.
However, the results on all the instances are considered in the figure.
According to these results, the two procedures present similar behaviors and
both outperform BL.

\begin{table}
\begin{center}
{\scriptsize
\begin{tabular}{cccccc|cc|cc|cc|c}
\cline{7-12}
\multicolumn{1}{l}{}       & \multicolumn{1}{l}{}  & \multicolumn{1}{l}{}   & \multicolumn{1}{l}{}   & \multicolumn{1}{l}{}       & \multicolumn{1}{l|}{}    & \multicolumn{2}{c|}{BL} & \multicolumn{2}{c|}{BBT} & \multicolumn{2}{c|}{BBS} & \multicolumn{1}{c}{} \\ \hline
\multicolumn{1}{|c|}{{\sc type}} & \multicolumn{1}{c|}{$w$} & \multicolumn{1}{c|}{$h$} & \multicolumn{1}{c|}{$n$} & \multicolumn{1}{c|}{delta} & \multicolumn{1}{c|}{{\scriptsize \#}{\sc i}}  & \multicolumn{1}{c|}{Resh} & \multicolumn{1}{c|}{Time} & \multicolumn{1}{c|}{Resh} & \multicolumn{1}{c|}{Time} & \multicolumn{1}{c|}{Resh} & \multicolumn{1}{c|}{Time} & \multicolumn{1}{c|}{Opt}   \\ \hline
\multicolumn{1}{|c|}{b} & \multicolumn{1}{c|}{6} & \multicolumn{1}{c|}{7} & \multicolumn{1}{c|}{29} & \multicolumn{1}{c|}{70\%} & \multicolumn{1}{c|}{40} & 17.580 & 0.00 & \textbf{17.180}$^\blacktriangle$ & 0.01 & 17.200 & 0.12 & \multicolumn{1}{c|}{17.180} \\ \cline{1-6} 
\multicolumn{1}{|c|}{b} & \multicolumn{1}{c|}{7} & \multicolumn{1}{c|}{6} & \multicolumn{1}{c|}{31} & \multicolumn{1}{c|}{75\%} & \multicolumn{1}{c|}{40} & 24.550 & 0.00 & \textbf{23.150}$^\blacktriangle$ & 0.02 & 23.180 & 0.18 & \multicolumn{1}{c|}{23.150} \\ \cline{1-6} 
\multicolumn{1}{|c|}{b} & \multicolumn{1}{c|}{7} & \multicolumn{1}{c|}{7} & \multicolumn{1}{c|}{31} & \multicolumn{1}{c|}{65\%} & \multicolumn{1}{c|}{40} & 20.030 & 0.00 & 19.480 & 0.04 & \textbf{19.430}$^\blacktriangle$ & 0.13 & \multicolumn{1}{c|}{19.430} \\ \cline{1-6} 
\multicolumn{1}{|c|}{b} & \multicolumn{1}{c|}{7} & \multicolumn{1}{c|}{7} & \multicolumn{1}{c|}{36} & \multicolumn{1}{c|}{75\%} & \multicolumn{1}{c|}{40} & 27.280 & 0.00 & \textbf{25.900} & 0.12 & \textbf{25.900} & 0.42 & \multicolumn{1}{c|}{25.850} \\ \cline{1-6} 
\multicolumn{1}{|c|}{u} & \multicolumn{1}{c|}{6} & \multicolumn{1}{c|}{7} & \multicolumn{1}{c|}{31} & \multicolumn{1}{c|}{75\%} & \multicolumn{1}{c|}{40} & 20.350 & 0.00 & \textbf{19.700}$^\blacktriangle$ & 0.02 & 19.730 & 0.06 & \multicolumn{1}{c|}{19.700} \\ \cline{1-6} 
\multicolumn{1}{|c|}{u} & \multicolumn{1}{c|}{7} & \multicolumn{1}{c|}{5} & \multicolumn{1}{c|}{21} & \multicolumn{1}{c|}{60\%} & \multicolumn{1}{c|}{40} & 15.300 & 0.00 & \textbf{14.730}$^\blacktriangle$ & 0.00 & 14.750 & 0.02 & \multicolumn{1}{c|}{14.730} \\ \cline{1-6} 
\multicolumn{1}{|c|}{u} & \multicolumn{1}{c|}{7} & \multicolumn{1}{c|}{7} & \multicolumn{1}{c|}{36} & \multicolumn{1}{c|}{75\%} & \multicolumn{1}{c|}{40} & 27.700 & 0.00 & 26.350 & 0.16 & \textbf{26.330} & 0.38 & \multicolumn{1}{c|}{26.300} \\ \hline 
\end{tabular}
}
\end{center}
\caption{Computational results on the set {\sc group 5}}
\label{tab:g5}
\end{table}

\begin{figure}[htbp] 
  \begin{center}

\includegraphics[width=6.5cm,height=5cm]{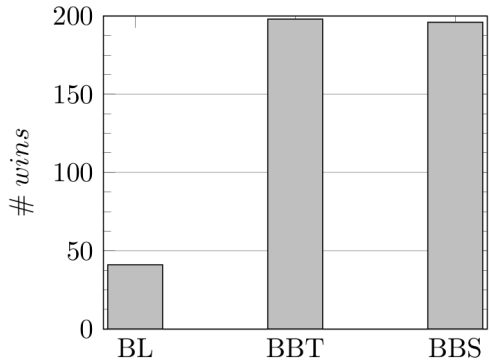}
    \end{center}
\caption{Aggregate results on the set {\sc group 5}}
\label{fig:g5}
\end{figure}

We present the computational results on the instances of {\sc group 6} in Table \ref{tab:g6} and
Figure \ref{fig:g6}.
In the table, according to the results in the literature, we report total reshuffles and times
instead of the average values.
Also in this case, the Bounded Beam Search heuristic outperforms all the other
approaches.

\begin{table}
\begin{center}
{\scriptsize
\begin{tabular}{ccc|cc|cc|cc|cc|cc|c}
\cline{4-11}
 & & \multicolumn{1}{c|}{} & \multicolumn{2}{c|}{IDA$^*$-R$^a$}  & \multicolumn{2}{c|}{BL$^b$} & \multicolumn{2}{c|}{BBT$^b$} & \multicolumn{2}{c|}{BBS$^b$} & \multicolumn{1}{c}{} \\ \hline
\multicolumn{1}{|c|}{$w$} & \multicolumn{1}{c|}{$h$} & \multicolumn{1}{c|}{{\scriptsize \#}{\sc i}}  & \multicolumn{1}{c|}{Resh} & \multicolumn{1}{c|}{Time} & \multicolumn{1}{c|}{Resh} & \multicolumn{1}{c|}{Time} & \multicolumn{1}{c|}{Resh} & \multicolumn{1}{c|}{Time} & \multicolumn{1}{c|}{Resh} & \multicolumn{1}{c|}{Time} & \multicolumn{1}{c|}{Opt}   \\ \hline
\multicolumn{1}{|c|}{6-10} & \multicolumn{1}{c|}{3} & \multicolumn{1}{c|}{1500} & \textbf{13592}$^\blacktriangle$ & 0.00  & 13603 & 0.00 & \textbf{13592}$^\blacktriangle$ & 2.94 & \textbf{13592}$^\blacktriangle$ & 6.01 & \multicolumn{1}{c|}{13592} \\ \cline{1-3} 
\multicolumn{1}{|c|}{6-10} & \multicolumn{1}{c|}{4} & \multicolumn{1}{c|}{2000} & \textbf{33095}$^\blacktriangle$ & 0.12 & 33535 & 0.10 & \textbf{33095}$^\blacktriangle$ & 6.10 & \textbf{33095}$^\blacktriangle$ & 13.30 & \multicolumn{1}{c|}{33095} \\ \cline{1-3} 
\multicolumn{1}{|c|}{6} & \multicolumn{1}{c|}{5} & \multicolumn{1}{c|}{500} & \textbf{9503}$^\blacktriangle$ & 0.05  & 9763 & 0.06 & \textbf{9503}$^\blacktriangle$ & 1.37 & \textbf{9503}$^\blacktriangle$ & 4.92 & \multicolumn{1}{c|}{9503} \\ \cline{1-3} 
\multicolumn{1}{|c|}{7} & \multicolumn{1}{c|}{5} & \multicolumn{1}{c|}{500}  & \textbf{11259}$^\blacktriangle$ & 0.16  & 11634 & 0.06 & \textbf{11259}$^\blacktriangle$ & 1.78 & \textbf{11259}$^\blacktriangle$ & 13.24 & \multicolumn{1}{c|}{11259} \\ \cline{1-3} 
\multicolumn{1}{|c|}{8} & \multicolumn{1}{c|}{5} & \multicolumn{1}{c|}{500}  & 12864 & 1.24  & 13232 & 0.08 & 12864 & 4.80 & \textbf{12863}$^\blacktriangle$ & 26.68 & \multicolumn{1}{c|}{12863} \\ \cline{1-3} 
\multicolumn{1}{|c|}{9} & \multicolumn{1}{c|}{5} & \multicolumn{1}{c|}{500} & \textbf{14158} & 2.74  & 14608 & 0.11 & \textbf{14158} & 10.81 & \textbf{14158} & 31.44 & \multicolumn{1}{c|}{14156} \\ \cline{1-3} 
\multicolumn{1}{|c|}{10} & \multicolumn{1}{c|}{5} & \multicolumn{1}{c|}{500} & \textbf{15727} & 6.91  & 16214 & 0.14 & 15728 & 21.69 & 15731 & 51.42 & \multicolumn{1}{c|}{15726} \\ \cline{1-3} 
\multicolumn{1}{|c|}{6} & \multicolumn{1}{c|}{6} & \multicolumn{1}{c|}{600} & \textbf{16175}$^\blacktriangle$ & 4.56  & 17084 & 0.11 & \textbf{16175}$^\blacktriangle$ & 5.06 & 16176 & 38.20 & \multicolumn{1}{c|}{16175} \\ \cline{1-3} 
\multicolumn{1}{|c|}{7} & \multicolumn{1}{c|}{6} & \multicolumn{1}{c|}{600} & 18607 & 10.72  & 19615 & 0.17 & 18606 & 21.44 & \textbf{18603} & 72.12 & \multicolumn{1}{c|}{18601} \\ \cline{1-3} 
\multicolumn{1}{|c|}{8} & \multicolumn{1}{c|}{6} & \multicolumn{1}{c|}{600}  & 21221 & 18.90 & 22385 & 0.25 & 21211 & 61.66 & \textbf{21195} & 96.70 & \multicolumn{1}{c|}{21187} \\ \cline{1-3} 
\multicolumn{1}{|c|}{9} & \multicolumn{1}{c|}{6} & \multicolumn{1}{c|}{600}  & 23790 & 36.72 & 25121 & 0.35 & 23793 & 158.69 & \textbf{23726} & 167.90 & \multicolumn{1}{c|}{$\ast$} \\ \cline{1-3} 
\multicolumn{1}{|c|}{10} & \multicolumn{1}{c|}{6} & \multicolumn{1}{c|}{600} & 26125 & 46.83 & 27518 & 0.47 & 26149 & 238.96 & \textbf{26025} & 238.13 & \multicolumn{1}{c|}{$\ast$} \\ \cline{1-3} 
\multicolumn{1}{|c|}{6} & \multicolumn{1}{c|}{7} & \multicolumn{1}{c|}{700}  & 24860 & 41.07  & 26694 & 0.25 & 24831 & 66.55 & \textbf{24827} & 136.22 & \multicolumn{1}{c|}{24813} \\ \cline{1-3} 
\multicolumn{1}{|c|}{7} & \multicolumn{1}{c|}{7} & \multicolumn{1}{c|}{700} & 28930 & 78.73  & 31124 & 0.41 & 28884 & 192.74 & \textbf{28835} & 192.75 & \multicolumn{1}{c|}{28771} \\ \cline{1-3} 
\multicolumn{1}{|c|}{8} & \multicolumn{1}{c|}{7} & \multicolumn{1}{c|}{700}  & 32782 & 102.69  & 35153 & 0.64 & 32701 & 404.96 & \textbf{32454} & 301.74 & \multicolumn{1}{c|}{$\ast$} \\ \cline{1-3} 
\multicolumn{1}{|c|}{9} & \multicolumn{1}{c|}{7} & \multicolumn{1}{c|}{700} & 37014 & 116.18  & 39483 & 0.90 & 37064 & 604.61 & \textbf{36494} & 349.37 & \multicolumn{1}{c|}{$\ast$} \\ \cline{1-3} 
\multicolumn{1}{|c|}{10} & \multicolumn{1}{c|}{7} & \multicolumn{1}{c|}{700} & 40896 & 129.75  & 43429 & 1.27 & 41053 & 732.33 & \textbf{40132} & 307.79 & \multicolumn{1}{c|}{$\ast$} \\ \hline 
\end{tabular}
}
\end{center}
{\scriptsize $^a$Pentium-IV 3.0 GHz and 1.0 GB RAM}\\
{\scriptsize $^b$ CPU Intel Core i7-3632QM 2.20GHz and 1.0 GB RAM}
\caption{Computational results on the set {\sc group 6}}
\label{tab:g6}
\end{table}

\begin{figure}[htbp] 
  \begin{center}
\includegraphics[width=7cm,height=6cm]{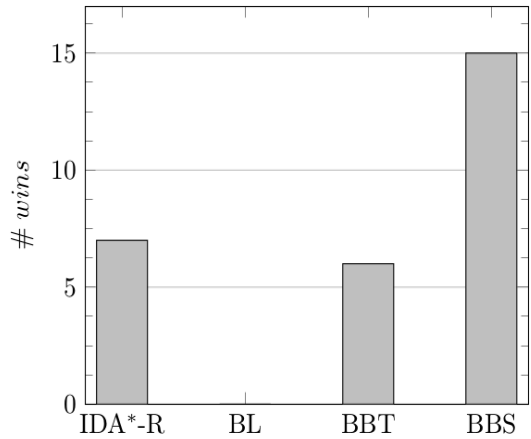}

    \end{center}
\caption{Aggregated results on the set {\sc group 6}}
\label{fig:g6}
\end{figure}

Finally, in Table \ref{tab:lbri} and Figure \ref{fig:lbri},
we report the results on the instances of {\sc lbri}.
Here, BBT can not find feasible solutions within the time limit on all the subsets of instances.
The Bounded Beam Search heuristic is again the best approach. 

Reviewing the whole computational experience, the Bounded Beam Search procedure
presents the best results on almost all instances
in the dataset. In particular, in the sets for which all the instances are solved to optimality, the Bounded Beam Search heuristic finds
such an optimum in more than 70\% of the cases. 

We think that such good performances are related to the combined use of upper and lower bound in the
selection of the most promising $\beta$ descendants of each node, in particular for breaking ties in the
evaluation.
Table~\ref{resultsheuristic}
emphasizes the dependency of the overall algorithm from the choice of an appropriate upper bound.

The second best algorithm is the one of ~\citet{TT2016}, when used as a heuristic by fixing a time limit of one second (BBT).
Nevertheless, it becomes impractical for large instances although the exact variant is extremely efficient for the
small ones.

\begin{table}
\begin{center}
{\scriptsize
\begin{tabular}{cccc|cc|cc|cc|}
\cline{5-10}
 & & & & \multicolumn{2}{c|}{BL$^a$} & \multicolumn{2}{c|}{BBT$^a$} & \multicolumn{2}{c|}{BBS$^a$} \\ \hline
\multicolumn{1}{|c|}{$w$} & \multicolumn{1}{c|}{$h$} & \multicolumn{1}{c|}{$n$} & \multicolumn{1}{c|}{{\scriptsize \#}{\sc i}}  & \multicolumn{1}{c|}{Resh} & \multicolumn{1}{c|}{Time} & \multicolumn{1}{c|}{Resh} & \multicolumn{1}{c|}{Time} & \multicolumn{1}{c|}{Resh} & \multicolumn{1}{c|}{Time}    \\ \hline
\multicolumn{1}{|c|}{50} & \multicolumn{1}{c|}{4}  & \multicolumn{1}{c|}{[196,199]} & \multicolumn{1}{c|}{400} & 99.755 & 0.004 & $\ast$ & $\ast$ & \textbf{98.803} & 0.111 \\ \cline{1-4} 
\multicolumn{1}{|c|}{50} & \multicolumn{1}{c|}{7}  & \multicolumn{1}{c|}{[343,349]} & \multicolumn{1}{c|}{700} & 271.676 & 0.000 & $\ast$ & $\ast$ & \textbf{261.684} & 1.003 \\ \cline{1-4} 
\multicolumn{1}{|c|}{50} & \multicolumn{1}{c|}{10}  & \multicolumn{1}{c|}{[490,499]} & \multicolumn{1}{c|}{1000} & 515.823 & 0.001 & $\ast$ & $\ast$ & \textbf{499.245} & 1.004 \\ \cline{1-4} 
\multicolumn{1}{|c|}{100} & \multicolumn{1}{c|}{4} & \multicolumn{1}{c|}{[396,399]} & \multicolumn{1}{c|}{400} & 196.313 & 0.024 & $\ast$ & $\ast$ & \textbf{195.145} & 0.379 \\ \cline{1-4} 
\multicolumn{1}{|c|}{100} & \multicolumn{1}{c|}{7} & \multicolumn{1}{c|}{[693,699]} & \multicolumn{1}{c|}{700} & 512.937 & 0.001 & $\ast$ & $\ast$ & \textbf{505.373} & 1.004 \\ \cline{1-4} 
\multicolumn{1}{|c|}{100} & \multicolumn{1}{c|}{10} & \multicolumn{1}{c|}{[990,999]} & \multicolumn{1}{c|}{1000} & 964.140 & 0.002 & $\ast$ & $\ast$ & \textbf{952.621} & 1.005 \\ \cline{1-4} 
\multicolumn{1}{|c|}{500} & \multicolumn{1}{c|}{4} & \multicolumn{1}{c|}{[1996,1999]} & \multicolumn{1}{c|}{400} & 964.168 & 0.116 & $\ast$ & $\ast$ & \textbf{963.383} & 0.843 \\ \cline{1-4} 
\multicolumn{1}{|c|}{500} & \multicolumn{1}{c|}{7} & \multicolumn{1}{c|}{[3493,3499]} & \multicolumn{1}{c|}{700} & 2338.576 & 0.008 & $\ast$ & $\ast$ & \textbf{2331.293} & 1.013 \\ \cline{1-4} 
\multicolumn{1}{|c|}{500} & \multicolumn{1}{c|}{10} & \multicolumn{1}{c|}{[4990,4999]} & \multicolumn{1}{c|}{1000} & 4277.015 & 0.014 & $\ast$ & $\ast$ & \textbf{4248.174} & 1.016 \\ \cline{1-4} 
\multicolumn{1}{|c|}{1000} & \multicolumn{1}{c|}{4} & \multicolumn{1}{c|}{[3996,3999]} & \multicolumn{1}{c|}{400} & 1922.303 & 0.014 & $\ast$ & $\ast$ & \textbf{1921.778} & 0.919 \\ \cline{1-4} 
\multicolumn{1}{|c|}{1000} & \multicolumn{1}{c|}{7} & \multicolumn{1}{c|}{[6993,6999]} & \multicolumn{1}{c|}{700} & 4561.887 & 0.029 & $\ast$ & $\ast$ & \textbf{4556.585} & 1.030 \\ \cline{1-4} 
\multicolumn{1}{|c|}{1000} & \multicolumn{1}{c|}{10} & \multicolumn{1}{c|}{[9990,9999]} & \multicolumn{1}{c|}{1000} & 8135.950 & 0.050 & $\ast$ & $\ast$ & \textbf{8109.102} & 1.037 \\ \hline
\end{tabular}
}
\end{center}
{\scriptsize $^a$ CPU Intel Core i7-3632QM 2.20GHz and 1.0 GB RAM}
\caption{Computational results on the set {\sc lbri}}
\label{tab:lbri}
\end{table}

\begin{figure}[htbp] 
  \begin{center}
\includegraphics[width=5.5cm,height=4cm]{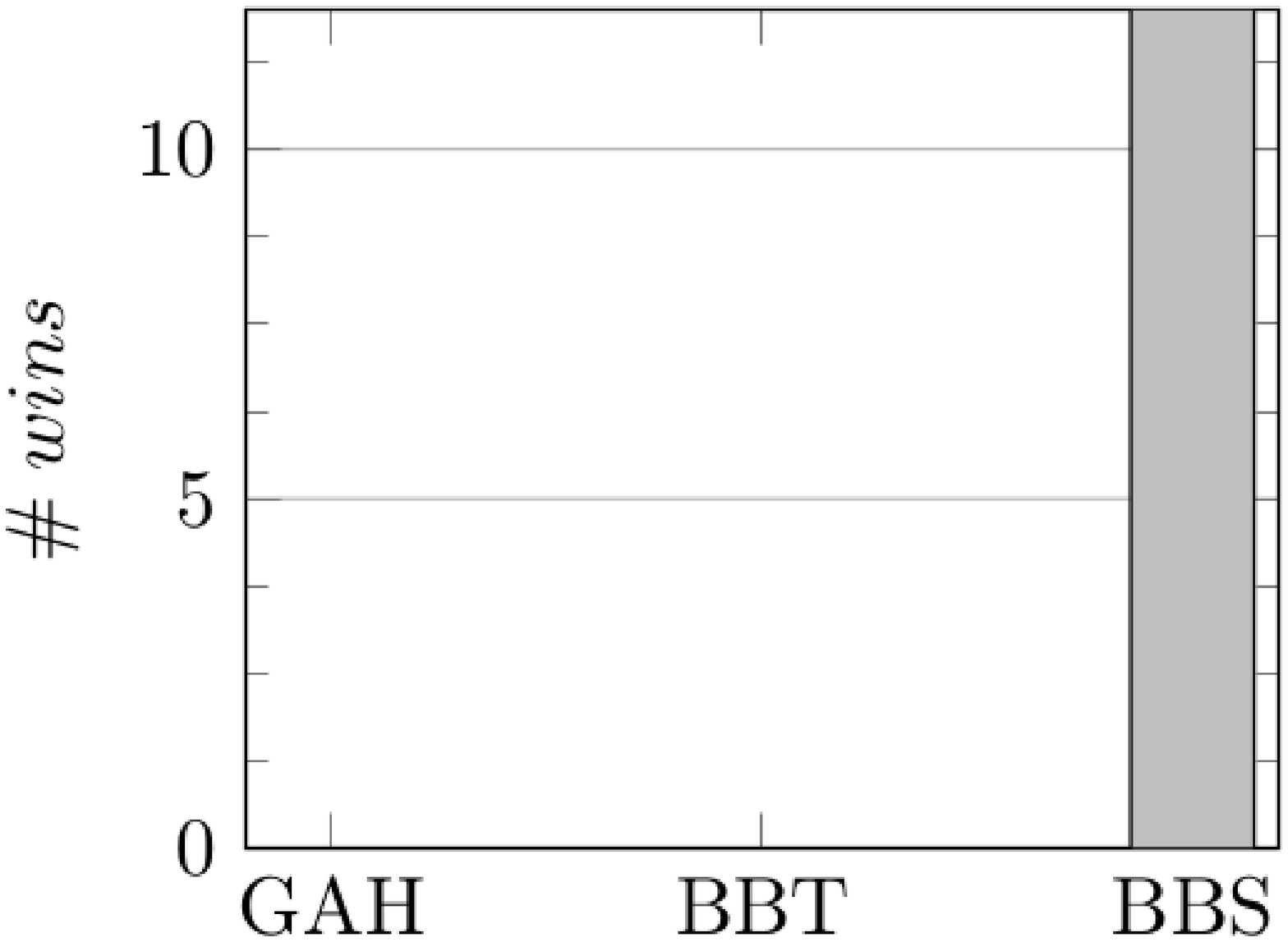}

    \end{center}
\caption{Aggregated results on the set {\sc lbri}}
\label{fig:lbri}
\end{figure}

\section{Conclusions}
\label{conclusions}

In this paper we considered the Block Relocation Problem, which has relevant practical applications in the container
logistics.
The problem is difficult both in theory and in practice and even the most effective exact approaches fail
to solve instances with size far away from the real ones.
In this paper we introduced a new lower bound and a new beam search heuristic for the problem.
The lower bound dominates most of the existing ones and it can be calculated by a polynomial time algorithm.
The beam search heuristic includes such procedure in order to reduce the search space and select the most promising nodes.
The results show that the proposed heuristic is very effective and outperforms all the other approaches on most of the considered
instances. 
The algorithm is also tested on a new set of real-size instances, {\sc lbri}, introduced here for the first time.

\section*{Acknowledgements}

The authors have been partially supported by Ministry of Instruction University
and Research (MIUR) with the program PRIN 2015,
project ``SPORT - Smart PORt Terminals'', code 2015XAPRKF;
project ``Nonlinear and Combinatorial Aspects of Complex Networks'',
code 2015B5F27W; project ``Scheduling cuts: new optimization models and
algorithms for cutting, packing and nesting in manufacturing processes'',
code 20153TXRX9; the CONtainer TRAnshipment STation (CONTRAST) 2013 project,
code FILAS-CR-2011-1315.

\newpage


\end{document}